\documentclass[a4paper,11pt]{article}
\pdfoutput=1
\usepackage[margin=1in]{geometry}
\usepackage{a4wide}
\usepackage[UKenglish]{babel}
\usepackage[noadjust]{cite}
\usepackage{amsmath}
\usepackage{amssymb}
\usepackage{amsthm} 
\usepackage{bbm}
\usepackage{mathtools}
\usepackage{mathrsfs}
\usepackage{csquotes}
\usepackage{todonotes}
\usepackage[title]{appendix}
\usepackage{xcolor}
\usepackage{blindtext}
\usepackage{enumerate}
\usepackage{manfnt}
\usepackage{multicol}
\usepackage{microtype}
\usepackage{amsmath,amssymb}
\usepackage{bm}

\usepackage{tikz}
\usetikzlibrary{calc}

\newcommand{\LC}[1]{\todo[size=footnotesize,bordercolor=black,color=green!40]{LC: #1}}

\MakeRobust{\ref}

\usepackage{stmaryrd}

\usepackage[noframe]{showframe}
\usepackage{framed}
\renewenvironment{shaded}{%
  \MakeFramed{\advance\hsize-\width \FrameRestore\FrameRestore}}%
 {\endMakeFramed}
\definecolor{shadecolor}{gray}{0.88}

\usepackage{mdframed}
\newmdtheoremenv[
    backgroundcolor=lightgray!20, 
    linecolor=black,
    innertopmargin=5pt,
    innerbottommargin=5pt,
    roundcorner=5pt,
]{theorem_highlighted}{Theorem}

\usepackage{dutchcal}

\usepackage[scr=boondox]{mathalfa}

\makeatletter
\newcommand{\reducedplus}{\mathpalette\reduced@plus\relax}
\newcommand{\reduced@plus}[2]{%
  \sbox6{$\m@th#1+$}%
  \sbox8{\scalebox{0.875}{\copy6}}%
  \dimen@=\dimexpr(\wd6-\wd8)/3\relax
  \raisebox{\dimen@}{\box8}%
}
\newcommand{\boxoperation}[2][\mathbin]{%
  #1{\mathpalette\box@operation{#2}}%
}
\newcommand{\box@operation}[2]{%
  \ooalign{$\m@th#1\boxempty$\cr\hidewidth$\m@th#1#2$\hidewidth\cr}%
}

\newcommand{\Eset}{E}

\DeclareSymbolFont{extraup}{U}{zavm}{m}{n}
\DeclareMathSymbol{\varheart}{\mathalpha}{extraup}{86}
\DeclareMathSymbol{\vardiamond}{\mathalpha}{extraup}{87}

\usepackage{hyperref} 
\hypersetup{colorlinks=true,citecolor=blue,linkcolor=blue,urlcolor=blue}

\newcommand{\Test}[2]{
\def\temp{#2}\ifx\temp\empty
  \operatorname{Test}_{#1}
\else
  \operatorname{Test}_{#1}^{#2}
\fi
}

\DeclarePairedDelimiter{\ceil}{\lceil}{\rceil}
\DeclarePairedDelimiter{\floor}{\lfloor}{\rfloor}

\newcommand{\dist}{\operatorname{dist}}
\newcommand{\diam}{\operatorname{diam}}
\newcommand{\ival}{\operatorname{ival}}
\newcommand{\isat}{\operatorname{isat}}
\newcommand{\Imm}{\operatorname{Im}}
\newcommand{\Sym}{\operatorname{Sym}}
\newcommand{\suchThat}{\mbox{ s.t. }}

\newcommand{\A}{\mathbf{A}}
\newcommand{\B}{\mathbf{B}}
\newcommand{\GG}{\mathbf{G}}
\newcommand{\K}{\mathbf{K}}
\newcommand{\T}{\mathbf{T}}
\newcommand{\X}{\mathbf{X}}
\newcommand{\Y}{\mathbf{Y}}
\newcommand{\W}{\mathbf{W}}

\newcommand{\N}{\mathbb{N}}

\newcommand{\F}{\mathbb{F}}

\newcommand{\HH}{\mathbb{H}}

\newcommand{\freeQ}{\mathbb{F}_{\Qminion}}

\renewcommand{\vec}[1]{\mathbf{#1}}
\newcommand{\ba}{\vec{a}}
\newcommand{\bb}{\vec{b}}

\newcommand{\bh}{\vec{h}}
\newcommand{\bi}{\vec{i}}

\newcommand{\bs}{\vec{s}}
\newcommand{\bu}{\vec{u}}
\newcommand{\bx}{\vec{x}}

\newcommand{\by}{\vec{y}}
\newcommand{\bw}{\vec{w}}

\newcommand{\sat}{\operatorname{sat}}

\newcommand{\qsat}{\operatorname{qsat}}

\newcommand{\val}{\operatorname{val}}

\newcommand{\qval}{\operatorname{qval}}

\newcommand{\dGames}{\mathbf{G}^{\mathbf{d\leftrightarrow d}}}

\DeclareMathAlphabet{\mathbx}{U}{BOONDOX-ds}{m}{n}
\SetMathAlphabet{\mathbx}{bold}{U}{BOONDOX-ds}{b}{n}
\DeclareMathAlphabet{\mathbbx} {U}{BOONDOX-ds}{b}{n}

\DeclareMathOperator{\tr}{tr}
\DeclareMathOperator{\Span}{span}

\DeclareMathOperator{\CSP}{CSP}

\DeclareMathOperator{\PVM}{PVM}

\DeclareMathOperator{\id}{id}

\DeclareMathOperator{\ar}{ar}

\DeclareMathOperator{\chiq}{\chi_{\operatorname{q}}}

\newcommand{\Qminion}{\ensuremath{{\mathscr{Q}}}}

\newcommand{\bzero}{\mathbf{0}}

\theoremstyle{plain}
\newtheorem{thm}{Theorem}
\newtheorem*{thm*}{Theorem}
\newtheorem{lem}[thm]{Lemma}
\newtheorem*{lem*}{Lemma}
\newtheorem{prop}[thm]{Proposition}
\newtheorem*{prop*}{Proposition}
\newtheorem{cor}[thm]{Corollary}
\newtheorem{conj}[thm]{Conjecture}

\theoremstyle{definition}
\newtheorem{defn}[thm]{Definition}
\newtheorem{rem}[thm]{Remark}

\newtheorem{example}[thm]{Example}

\newtheorem{claim}{Claim}[thm]

\newtheorem{subclaim}{Claim}[claim]

\newtheorem{fact}{Fact}[thm]

\begin{document}

\author{Lorenzo Ciardo\\
University of Oxford\\
\texttt{lorenzo.ciardo@cs.ox.ac.uk}
}

\title{
On the quantum chromatic gap
}
\date{}

\maketitle

\begin{abstract}
\noindent 
The largest known gap between quantum and classical chromatic number of graphs, obtained via quantum protocols for colouring Hadamard graphs based on the Deutsch--Jozsa
algorithm and the quantum Fourier transform, is exponential.
We put forth a quantum pseudo-telepathy version of Khot's $d$-to-$1$ Games Conjecture and prove that, conditional to its validity, the gap is unbounded: There exist graphs whose quantum chromatic number is 3 and whose classical chromatic number is arbitrarily large. 
Furthermore, we show that the existence of a certain form of pseudo-telepathic \textsc{XOR} games would imply the conjecture and, thus, the unboundedness of the quantum chromatic gap. As two technical steps of our proof that might be of independent interest, we establish a quantum adjunction theorem for Pultr functors between categories of  relational structures, and we prove that the Dinur--Khot--Kindler--Minzer--Safra reduction, recently used for proving the $2$-to-$2$ Games Theorem, is quantum complete.
\end{abstract}

\section{Introduction}
Take a graph $\GG$ and $n$ distinct colours. Try to assign a colour to each of the vertices of $\GG$ in a way that no edge is monochromatic. What is the minimum $n$ required for completing the task? The answer to this question is a classic graph-theoretic parameter, studied from the very beginning of the history of graph theory: the chromatic number of $\GG$ (in symbols, $\chi(\GG)$).

Consider now a one-round game involving two cooperating, non-communicating, computationally unbounded players (Alice and Bob) and a verifier. The goal of Alice and Bob is to convince the verifier that $\GG$ can be coloured using $n$ colours. The verifier questions the players by selecting two vertices $v_A$ and $v_B$ of $\GG$ at random and sending $v_A$ to Alice and $v_B$ to Bob. Alice replies with a colour $c_A$, and Bob replies with a colour $c_B$. The game is won if the answers are consistent with each other (i.e., $c_A=c_B$ when $v_A=v_B$) and compatible with a proper colouring of $\GG$ (i.e., $c_A\neq c_B$ when $v_A$ and $v_B$ are adjacent). What is the minimum $n$ that allows Alice and Bob to win the colouring game, no matter which questions the verifier asks? A moment of thought reveals that the answer is again $\chi(\GG)$: If $\GG$ is $n$-colourable (i.e., if $n\geq \chi(\GG)$), a winning strategy for Alice and Bob can be obtained by choosing a colouring before the game start, and always replying according to that colouring. If $\GG$ is not $n$-colourable (i.e., if $n<\chi(\GG)$), any strategy can be broken by some pair of questions chosen by the verifier.
\paragraph{Quantum pseudo-telepathy}
Describing the chromatic number in the language of multiplayer games / multiprover interactive proofs has the advantage of turning the original, purely combinatorial question into a question about information---in particular, about the limits imposed by the physical model we choose for describing information transfer. At a high level, the reason why Alice and Bob cannot find a perfect strategy using fewer colours than $\chi(\GG)$ is that the lack of communication in the game sets strong restrictions on the possible correlations between the players' answers. To make these restrictions explicit, imagine to build an experimental setup for the colouring game.
Before the game starts, Alice and Bob are free to communicate, which means that they have complete access to a shared amount of information, call it $\mathcal{I}$. 
When the game starts, to enforce no communication, we make sure that the distance separating them is big enough (and the questions and replies are sent simultaneously in at least one reference frame). This means that now Alice has only access to one part $\mathcal{I}_A$ of $\mathcal{I}$, and Bob has only access to the other part $\mathcal{I}_B$.
Upon receiving question $v_A$, Alice aims to read her answer $c_A$ from $\mathcal{I}_A$; physically, this can be expressed by Alice \textit{measuring} the state of $\mathcal{I}_A$ via some measurement depending on $v_A$. Bob does the same: He measures $\mathcal{I}_B$ according to the question $v_B$, and sends the measurement outcome $c_B$ to the verifier.
This setup precisely captures the intended rules of the colouring game in a theory of physics that is \textit{realistic} (the outcome of a measurement is completely determined by the description of the measurement and the physical properties of the system being measured) and \textit{local} (Alice's choice of measurement does not impact Bob's outcome, and vice versa).
Yet, if we were to actually implement the experiment, we might be surprised at finding that, for some choices of $\GG$, Alice and Bob are able to convince the verifier that $\GG$ can be coloured with \textit{strictly fewer} colours than $\chi(\GG)$, thus effectively violating the restrictions imposed by local realism. Indeed, by
$(i)$ encoding the shared information $\mathcal{I}$ in a bipartite quantum system, $(ii)$ preparing the system in a suitably entangled state, and $(iii)$ performing on each of the two parts of the system carefully chosen positive-operator valued (or, without loss of generality, projection valued) measurements, the players' answers can exhibit non-classical correlations that allow them to win with certainty the classically unwinnable game. This phenomenon is at the heart of Bell's Theorem~\cite{bell1964einstein} on the incompatibility of quantum theory with local hidden-variable theories, and it became known as \textit{quantum pseudo-telepathy}~\cite{brassard2005quantum_pseudo,brassard1999cost}. In fact, it is an example of a so-called \textit{proof of Bell's Theorem without inequalities} (other famous such examples being the three-player GHZ game~\cite{greenberger1989going,greenberger1990bell} and the two-player Mermin ``magic-square'' game~\cite{mermin1990simple,mermin1993hidden,peres1990incompatible}), in that the best quantum strategy not only outperforms all classical strategies, but it also allows the players to win with certainty---unlike other examples of violations of Bell's inequality such as the CHSH game~\cite{clauser1969proposed}.

\paragraph{Correspondences}
Let $\chiq(\GG)$ be the minimum number $n$ of colours required for two entangled players to win the colouring game on $\GG$ with certainty; in~\cite{CameronMNSW07}, this parameter was named the \textit{quantum chromatic number} of $\GG$. Since quantum strategies are at least as powerful as classical strategies, it is clear that the inequality $\chiq(\GG)\leq\chi(\GG)$ holds for each $\GG$.
The problem of understanding the maximum  possible gap between quantum and classical chromatic number of graphs was already considered in~\cite{BuhrmanCW98,brassard1999cost,de2001quantum,buhrman2001quantum}. 
The first example of a graph for which the inequality is strict was exhibited in~\cite{brassard1999cost}, with the description of a winning quantum strategy for colouring the vertices of certain Hadamard graphs based on the Deutsch--Jozsa
protocol~\cite{deutsch1992rapid} using fewer colours than the classical chromatic number, see also~\cite{BuhrmanCW98}. The result was later extended to all Hadamard graphs in~\cite{AvisHKS06}, via a different protocol based on the quantum Fourier transform~\cite{coppersmith2002approximate,mosca2004exact}. In particular, in combination with the known asymptotic behaviour of the classical chromatic number of Hadamard graphs~\cite{frankl1987forbidden}, those results imply the existence of a class of graphs exhibiting an \underline{exponential} gap between quantum and classical chromatic number; i.e., $\chi=\Omega(2^{\chi_q})$.
Another line of work focussed on producing smaller and smaller examples of graphs for which $\chi_{\operatorname{q}}<\chi$. The number of vertices of the smallest explicit instance of a graph whose quantum and classical chromatic numbers differ was progressively decreased from 32768~\cite{galliard2003impossibility}, to 1609~\cite{AvisHKS06}, to 18~\cite{CameronMNSW07}, to 14~\cite{MancinskaR16_oddities} (see also related work in~\cite{scarpa2011kochen,elphick2019spectral,lalonde2023quantum,godsil2024quantum,abiad2025eigenvalue}).
The perfect quantum strategies for those graphs are mostly based on real
or complex orthogonal representations of graphs, with the notable exception of~\cite{MancinskaR16_oddities}.
Despite extensive work on this problem, it is currently open whether the $\chi_{\operatorname{q}}$ vs. $\chi$ gap can be larger than exponential.


Let us now set aside the question on the quantum chromatic gap for a moment, and turn to a problem from classical complexity theory. Take a graph $\GG$ and $n$ distinct colours. What is the computational complexity of finding an explicit $n$-colouring of $\GG$, assuming that one exists? Unless $n\leq 2$, in which case a simple polynomial-time algorithm works, \textit{$n$-colouring} is a classic NP-hard computational problem, which already appeared
on Karp's original list~\cite{Karp72}. It is then natural to consider the following approximate version of it: Assuming that $\GG$ is $n$-colourable, find an explicit $n'$-colouring, for some $n'>n$. A variant of this problem was already considered in~\cite{GJ76}, where NP-hardness was conjectured to hold for every $3\leq n\leq n'$. Despite decades of attempts, the complexity of the \textit{approximate graph colouring} problem (\textsc{AGC}) is still wide open in general. For example, the case
$n=3$, $n'=4$ was only proved to be NP-hard in~\cite{khanna1993hardness_conference}, see also~\cite{GK04} for a simpler proof. More generally,~\cite{KhannaLS00} showed hardness of the case
$n'=n+2\floor{n/3}-1$. This was improved to $n'=2n-2$ in~\cite{BrakensiekG16}, and recently to $n'=2n-1$ in~\cite{BBKO21}. In particular, this last result implies
hardness of the case $n=3$, $n'=5$---while even the complexity of the case $n=3$, $n'=6$ is
still open. The asymptotic state-of-the-art hardness is \underline{exponential}: Building on~\cite{Khot01,Huang13}, NP-hardness of the case $n'={n\choose\floor{n/2}}-1$ was established in~\cite{KOWZ22}.

\paragraph{CSP approximation}

Graph colouring is an example of \textit{constraint satisfaction problem} (CSP)---i.e., the problem of assigning one of a given set of labels (for example, colours) to each of a given set of variables (for example, vertices of a graph) in a way that a list of constraint predicates (for example, ``adjacent vertices should be assigned distinct colours'') is satisfied. Other typical examples of CSPs are Boolean satisfiability problems, graph and hypergraph homomorphism problems, and linear equations. The complexity landscape of the \textit{exact} version of CSPs has been explored, at this level of generality, since the foundational work~\cite{Feder98:monotone} (while specific fragments have been investigated since the dawn of complexity theory, for example~\cite{Karp72, Schaefer78:stoc,HellN90}). In particular, a line of work initiated in~\cite{Jeavons98:algebraic,Jeavons97:closure} investigated the complexity of CSPs in relation to \textit{polymorphisms}---invariance properties of the algebraic structures encoding the CSP constraints. This approach eventually culminated with a proof of the CSP Dichotomy Theorem~\cite{Bulatov17:focs,Zhuk17_FOCS,Zhuk20:jacm}, establishing a P vs. NP-hard dichotomy for the whole class of finite-domain CSPs.

In the \textit{approximation} regime, the complexity of CSPs is far less understood than in the exact regime. Many of the known approximation hardness results were derived from the \textit{Probabilistically Checkable Proofs (PCP)} Theorem~\cite{Arora98:jacm-proof,Arora98:jacm-probabilistic,feige1996interactive,Dinur07:jacm}.
In particular, H\aa stad derived from the PCP Theorem a celebrated result on the \textit{hardness of approximating linear equations}~\cite{haastad2001some}: For every constant $\epsilon>0$, it is NP-hard to find an assignment for a system of Boolean equations having, say, three variables per equation (i.e., an instance of \textsc{3XOR}) satisfying $\frac{1}{2}+\epsilon$ fraction of the equations, even when the system is promised to satisfy $1-\epsilon$ fraction of the equations. This inapproximability gap is optimal: On the completeness side, increasing $1-\epsilon$ to $1$ (i.e., requiring perfect completeness) would make the problem tractable in polynomial time via Gaussian elimination; on the soundness side, assigning a random value to each variable satisfies half of the equations in expectation and, thus, gives a $\frac{1}{2}$-approximation algorithm.

While the PCP Theorem provided optimal inapproximability thresholds for various well-studied problems (in the context of CSPs and beyond), including 
$\textsc{Set Cover}$~\cite{feige1998threshold}, $\textsc{3SAT}$~\cite{haastad2001some}, $\textsc{Clique}$~\cite{hastad1996clique}, and the aforementioned \textsc{3XOR}, for many other fundamental computational problems the gap between the best known tractability factor and the
best known hardness factor could not be decreased to zero. This was the main motivation behind Khot's introduction of the
\textit{Unique Games Conjecture (UGC)} ~\cite{Khot02stoc}, postulating that  a special type of \textit{label cover} CSPs having $1$-to-$1$ (i.e., bijective) constraints are hard to approximate. By exhibiting suitable reductions from $\textsc{Unique Games}$, the precise threshold of approximation hardness for a variety of flagship computational problems in theoretical computer science was found conditional to the truth of the UGC, primary examples being
$\textsc{Max-Cut}$~\cite{khot2007optimal},
$\textsc{Vertex Cover}$~\cite{khot2008vertex}, and $\textsc{Max Acyclic Subgraph}$~\cite{guruswami2008beating}.
The inherent imperfect
completeness of the UGC makes it unsuitable as a basis for conditional hardness results concerning \textsc{AGC}. On the other hand, the $d$-to-$1$ version of the conjecture---also introduced in~\cite{Khot02stoc}---has been extensively investigated as a source of hardness for \textsc{AGC}. In particular,~\cite{Dinur09:sicomp} showed hardness of $\mathcal{O}(1)$-colouring a $4$-colourable graph under the $2$-to-$1$ Conjecture.
In addition, hardness in the $3$ vs. $\mathcal{O}(1)$ regime and in a $4$ vs. superconstant regime was obtained in~\cite{Dinur09:sicomp} and~\cite{dinur2010conditional}, respectively, conditional to certain non-standard variants of the conjecture.
Later,~\cite{GS20:icalp} proved hardness of $3$ vs. $\mathcal{O}(1)$ colouring under the $d$-to-$1$ Conjecture for any fixed $d$. 

Recently, a substantial step towards the proof of the UGC was done in the sequence of works~\cite{Khot17:stoc-independent,Khot18:focs-pseudorandom,Dinur18:stoc-non-optimality,Dinur18:stoc-towards}. Using H\aa stad's inapproximability of \textsc{3XOR} as the \textit{source of hardness}, the $2$-to-$2$ Conjecture (with imperfect completeness) was proved, which immediately yields ``half'' of the UGC: Given a $(\frac{1}{2}-\epsilon)$-satisfiable instance of $\textsc{Unique Games}$, it is NP-hard to find an
assignment satisfying an $\epsilon$-fraction of the constraints. The result is obtained via an advanced analysis of the expansion properties of Grassman graphs arising from a reduction from \textsc{3XOR}.
While the imperfect completeness of the result does not allow deducing from it hardness of \textsc{AGC} via~\cite{Dinur09:sicomp,GS20:icalp}, it yields hardness for a variant of \textsc{AGC} where one aims to ``almost colour'' (i.e., to find a proper colouring for a large subgraph of the instance graph) almost 3-colourable graphs~\cite{HechtMS23}.

\paragraph{Contributions}
We prove that \textit{the gap between quantum and classical chromatic number of graphs is unbounded, conditional to a hypothesis on the pseudo-telepathy of $d$-to-$1$ label cover and \textsc{3XOR} games}. Our result is obtained by translating certain machinery developed in the context of CSP approximation hardness into the setting of quantum pseudo-telepathy.  
For the first part of our result, we consider known pseudo-telepathic non-local games encoding $d$-to-$1$ CSPs, and we develop a framework for using them as a  \textit{source of pseudo-telepathy} for different games---in particular, for the colouring game. To that end, we prove a quantum version of adjunction for functors between relational structures known as \textit{Pultr functors}~\cite{pultr_adjoints}. Via this adjunction, we obtain soundness and completeness of a large class of reductions used in the context of CSP approximation---including all gadget reductions as well as reductions that are provably not encoded by gadgets.
In this way, we derive an optimal $3$ vs. $\mathcal{O}(1)$ quantum chromatic gap conditional to a \textit{strengthened} version of pseudo-telepathy for $d$-to-$1$ games, involving measurement scenarios that are \textit{locally compatible}, meaning that the measurements corresponding to bounded-size subinstances of the CSP can be performed simultaneously.
The second part of the result gives a possible avenue for proving strong pseudo-telepathy of $2$-to-$1$ games (and, thus, unbounded quantum chromatic gap) by amplifying the pseudo-telepathy of \textsc{3XOR} games~\cite{mermin1990simple,mermin1993hidden,peres1990incompatible}. To do that, we prove that the reduction constructed in~\cite{Khot17:stoc-independent,Khot18:focs-pseudorandom,Dinur18:stoc-non-optimality,Dinur18:stoc-towards} to establish the $2$-to-$2$ Games Theorem is quantum complete. While in the classical approximation-hardness setting the result holds with \textit{imperfect} completeness (as the exact version of \textsc{3XOR}, the source of hardness, is solvable in polynomial time), in the quantum pseudo-telepathy setting \textit{perfect} completeness is preserved.

\paragraph{Related work}

It was shown in~\cite{ciardo_quantum_minion} that the occurrence of pseudo-telepathy in non-local games is governed by the polymorphisms of the corresponding CSPs, thereby connecting pseudo-telepathy to \textit{exact} CSP solvability. The current work extends the connection to the \textit{approximation} regime.
%
Several recent works have made use of classical CSP reductions to investigate
the computational complexity of the quantum version of classes of CSPs, see e.g.~\cite{culf2024approximation,culf2024re,mousavi2025quantum}.
In particular, a reduction from \textsc{$3$SAT} to $3$-colouring preserving quantum homomorphisms was considered in~\cite{fukawa2011quantum,ji2013binary}. This was recently generalised in~\cite{harris2024universality}, which described a reduction from arbitrary non-local games to the quantum $3$-colouring game. 
These reductions, based on so-called commutativity gadgets, are not sound in the gap regime: Starting from a highly unsatisfiable instance, they do not guarantee high classical chromatic number of the output graph, and thus cannot be used to produce graphs exhibiting a large $\chiq$ vs. $\chi$ gap. 
It was proved in~\cite{AtseriasKS19} that gadget reductions for Boolean CSPs preserve the existence of perfect unitary operator assignments. The result was recently extended to the non-Boolean case in~\cite{bulatov2024satisfiabilityV3}, where it is mentioned that it does not apply to quantum CSP assignments, as considered in~\cite{MancinskaR16,ciardo_quantum_minion} (as well as~\cite{abramsky2017quantum} and the current work), since the latter are defined in terms of idempotent operators (projectors). 
It was recently proved in~\cite{simul_banakh} that the ratio between classical and quantum chromatic number of a graph is upper bounded by a function that depends only on the dimension of the Hilbert space describing the shared quantum system used in the quantum strategy.
Non-local games exhibiting a large separation between classical and quantum value (often called \textit{strong violations of Bell's inequality} in the literature) have been investigated in quantum physics, due to their applicability to demonstrate the presence of quantum entanglement, see e.g.~\cite{werner2001bell}. 
The existence of infinite families of games whose classical value approaches $0$ despite their quantum value being $1$ follows from~\cite{bell1964einstein,raz1995parallel}. Instances of unique (i.e., $1$-to-$1$) games with quantum value arbitrarily close to $1$ and classical value arbitrarily close
to $0$ can be obtained from~\cite{kempe2010unique} in combination with the SDP integrality gaps from~\cite{khot2015unique}. A weaker violation can be similarly obtained for $d$-to-$d$ games via~\cite[Theorem~4.8]{kempe2010unique}, see also~\cite{dinur2015parallel}.

\section{Overview of results and techniques}
In this section, we discuss the main results of this paper and highlight the ideas needed for their proof. 
Most of the concepts appearing in this section are described here at the intuitive level.
The formal definitions, as well as the complete proofs of the results, are given in the subsequent sections.
\subsection{The source of pseudo-telepathy}
\label{subsec_classical_quantum_CSPs}
Traditionally, several hardness of approximation results have been established via the proof composition paradigm introduced by
Arora--Safra~\cite{Arora98:jacm-probabilistic}. At a high level, this consists in composing an
``outer verifier''---based on a prototypical NP-hard computational problem such as $\textsc{Gap-3SAT}$ or $\textsc{Gap-3XOR}$, whose existence is derived from the PCP Theorem---with an ``inner verifier''---often based on modifications of the long code,
allowing to transfer approximation hardness through a sequence of transformations of the initial computational problem.
Our first step is to build a suitable source from which we can transfer pseudo-telepathy up to the colouring game. In analogy to the approximation setting, we choose a particular type of constraint satisfaction problem known as $d$-to-$1$ label cover as the source of pseudo-telepathy.

\paragraph{Classical setting}
An instance $\Phi$ of constraint satisfaction problem (CSP) consists of a set of variables, a set of local constraints on tuples of variables, and a set of weights over the constraints. The task is to assign to each variable a value from some fixed alphabet, with the goal of maximising the total weight of satisfied constraints.
The maximum weight of simultaneously satisfiable constraints is the classical value $\sat(\Phi)$ of $\Phi$. We say that $\Phi$ is a $d$-to-$1$ label cover instance for some $d\in\N$ if the constraints $(i)$ are binary, $(ii)$ form a bipartite graph, and $(iii)$ are $d$-to-$1$, in the sense that, given a constraint on variables $(x,y)$, for each labelling of $x$ there is a unique admitted labelling of $y$ and for each labelling of $y$ there are $d$ admitted labellings of $x$.

Checking whether a $d$-to-$1$ instance $\Phi$ admits a perfect classical assignment (i.e., checking whether $\sat(\Phi)=1$) is easily seen to be NP-hard unless $d=1$ (in which case, the instance is known as a unique game).
Khot famously conjectured that the soundness can be reduced arbitrarily while preserving NP-hardness. 

\begin{conj}[$d$-to-$1$ Conjecture, \cite{Khot02stoc}]
\label{conj_d_to_1_khot}
Fix $d\geq 2$. For each $\epsilon>0$ there exists an integer $n$ such that, for a $d$-to-$1$ instance $\Phi$ on alphabet size $n$, it is NP-hard to decide whether $\sat(\Phi)=1$ or $\sat(\Phi)<\epsilon$.
\end{conj}

As mentioned in the Introduction, the $d$-to-$1$ Conjecture---as well as the Unique Games Conjecture, corresponding to the case $d=1$ with imperfect completeness---have been used in the last two decades as the source of several optimal inapproximability results for CSPs.

\paragraph{Quantum setting}
We now introduce quantum assignments in the picture. 
Just like the colouring game discussed in the Introduction, any CSP instance $\Phi$ can be expressed as a two-player one-round game involving non-communicating, cooperating, computationally unbounded players and a verifier. The existence of a perfect classical strategy for the game corresponds to $\Phi$ being perfectly satisfiable. However, when the players are allowed to share a bipartite quantum system as part of the strategy and base their replies on measurements of their parts of the system, 
the answers can exhibit stronger-than-classical correlations that might give rise to perfect strategies even when $\sat(\Phi)<1$---a phenomenon known as pseudo-telepathy.

In particular, the non-local game corresponding to $d$-to-$1$ instances is pseudo-telepathic; in fact, the same holds for the game associated with any NP-hard CSP, as was shown in~\cite{ciardo_quantum_minion}.
In analogy with Khot's Conjecture~\ref{conj_d_to_1_khot}, we conjecture the existence of $d$-to-$1$ games that can be won with certainty by quantum players, but for which the best winning probability for classical players is arbitrarily close to $0$. 
Furthermore, we conjecture that the quantum strategy can be taken to be \textit{locally compatible}, in a sense that we now discuss.

As formally described in Section~\ref{sec_classical_quantum_CSPs}, a quantum strategy corresponds to a collection of measurements indexed by the variables of the instance $\Phi$. Imagine, following~\cite[\S.7]{helton2017algebras}, that each variable is a laboratory where the players conduct experiments upon receiving their questions from the verifier; the outcomes of the experiments are the answers sent to the verifier. In quantum mechanics, commuting observables can be simultaneously measured. Hence, if two laboratories are sufficiently close in the metric associated with the instance (say, their distance in the Gaifman graph of $\Phi$ is at most $k$), we may require that a joint experiment could be conducted in the two laboratories---and, thus, that the corresponding measurements commute. This gives rise to a notion of \textit{$k$-compatible} quantum strategies,\footnote{This should not be confused with quantum versions of spoiler-duplicator tests
for CSPs (which capture the bounded-width or local-consistency algorithm in the classical setting), as recently introduced in~\cite{thesis_amin_karamlou}.} whose projectors                are simultaneously diagonalisable on subinstances of $\Phi$ of diameter $k$. In particular, for $k=0$, this notion does not enforce any commutativity requirement, as is commonly done in the context of binary games, see for example~\cite{MancinskaR16,kempe2010unique}. The case $k=1$ corresponds to enforcing commutativity for the measurements associated with variables lying in a common constraint, which is typically required for games having non-binary constraints~\cite{abramsky2017quantum}. The case $k=2$ has appeared in~\cite{mousavi2025quantum}, where commutativity between measurements associated with paths of length $2$ was required to derive completeness of a reduction between the quantum versions of $\textsc{Unique Games}$ and $\textsc{MaxCut}$ (such strategies were called \textit{weak quantum assignments} therein, see~\cite[Definition~22]{mousavi2025quantum}). At the other extreme, when $k$ is the diameter of the Gaifman graph associated with $\Phi$, the requirement corresponds to simultaneous diagonalisability of all measurements and, thus, it implies a collapse to classical strategies.
We propose the following pseudo-telepathy version of Conjecture~\ref{conj_d_to_1_khot}.

\begin{shaded}
\vspace{-.4cm}
\begin{conj}
\label{conj_d_to_1_khot_quantum}
Fix $d\geq 2$. For each $\epsilon>0$ and each $k\in\N$ there exists a $d$-to-$1$ instance $\Phi$ (on some alphabet size $n$) such that
$\Phi$ admits a perfect $k$-compatible quantum assignment but $\sat(\Phi)<\epsilon$.
\end{conj}
\vspace{-.4cm}
\end{shaded}

\subsection{Quantum adjunction of Pultr functors}
\label{subsec_pultr_functors}
Having fixed the outer verifier, the next step is to see how to transfer pseudo-telepathy up to the colouring game described in the Introduction. To that end, we make use of categorical operators known as \textit{Pultr functors} between categories of relational structures~\cite{pultr_adjoints}. These are pairs of transformations $(\Lambda,\Gamma)$ between relational structures on possibly different signatures, parameterised by a fixed \textit{Pultr template} (see the formal definitions in Section~\ref{sec_pultr_functors}). The functor $\Lambda$ (known as the \textit{left} Pultr functor) acts on a structure $\X$ by replacing its vertices with copies of the structures in the Pultr template, and then identifying some of the vertices. The functor $\Gamma$ (the \textit{central} Pultr functor) acts on a structure $\Y$ by replacing its vertices with \textit{homomorphisms} between the structures in the Pultr template and $\Y$.

\begin{example}
\label{ex_some_pultr_functors}
To illustrate the type of combinatorial operations that can be expressed as Pultr functors at the intuitive level, we now describe two examples of pairs $(\Lambda,\Gamma)$ of Pultr functors, both defined on the category of directed graphs (i.e., relational structures having a unique, binary relation).
Fix a digraph $\GG$. The first pair $(\Lambda',\Gamma')$ of Pultr functors is defined as follows: For given digraphs $\X$ and $\Y$, $\Lambda'\X$ is the direct product $\X\times\GG$, while $\Gamma'\Y$ is the exponential digraph $\Y^{\GG}$ (for the definitions, we refer the reader to~\cite[\S2.4]{hell2004graphs}).
The second pair $(\Lambda'',\Gamma'')$ of Pultr functors is defined as follows: For given digraphs $\X$ and $\Y$, $\Lambda''\X$ is the digraph obtained by replacing each vertex $x$ of $\X$ with a directed edge $(x^{(1)},x^{(2)})$, and then identifying $x^{(2)}$ with $y^{(1)}$ for each edge $(x,y)$ of $\X$; $\Gamma''\Y$ is the digraph whose vertex set is the edge set of $\Y$, and whose edge set contains all pairs of edges of the form $((x,y),(y,z))$.
%
    %
It is not hard to check that both pairs of functors satisfy the relation $\Lambda\X\to\Y$ $\Leftrightarrow$ $\X\to\Gamma\Y$ for each pair $\X,\Y$ of digraphs. As we shall see, this holds in general for every pair of Pultr functors. 
\end{example}

The benefit of considering such an abstract type of transformations is that they are general enough to capture many classes of reductions used in the theory of CSPs. In particular, \textit{gadget reductions} can be viewed as left Pultr functors for specific Pultr templates. 
Crucially, various CSP reductions that are \textit{not} gadget reductions are also captured via Pultr functors, a primary example being the \textit{line-digraph} reduction---a standard graph-theoretic operation~\cite{harary1960some} corresponding to the functor $\Gamma''$ in Example~\ref{ex_some_pultr_functors}.
The reason why Pultr functors yield reductions between classical CSPs is that they satisfying a (thin) adjunction property. 
In the next statement, $\rho$ and $\tau$ are the relational signatures associated with the Pultr template, while the notation $\A\to\B$ denotes the existence of a \textit{homomorphism} between $\A$ and $\B$ (which, in this case, need to have the same signature). 
\begin{thm}[\cite{pultr_adjoints}]
\label{thm_pultr_adjoints}
    Let $\mathfrak{T}$ be a $(\rho,\tau)$-Pultr template. For any $\tau$-structure $\X$ and any $\rho$-structure $\Y$ it holds that
$\Lambda\X\to\Y$ if and only if $\X\to\Gamma\Y$.
\end{thm}

In order to use Pultr functors for transferring pseudo-telepathy across non-local games, we need a quantum version of the adjunction theorem stated above.
First of all, let us introduce some notation. Any two structures $\A$ and $\B$ on the same signature can be viewed as a CSP instance $\Phi$ where the variables are the vertices of $\A$, the labels in the alphabet are the vertices of $\B$, the constraints correspond to tuples in the relations of $\A$, and the admitted assignments to the constraints are expressed by tuples in the relations of $\B$.\footnote{Since, in this context, we only deal with perfect strategies, the constraints can be weighted according to any positive probability distribution.} In this way, $\sat(\Phi)=1$ precisely when $\A\to\B$. We indicate the existence of a perfect $k$-compatible quantum strategy for such CSP instance by the notation $\A\leadsto^k\B$. The next result shows that one direction of Theorem~\ref{thm_pultr_adjoints} can be quantised, provided that the Pultr template is connected (in a technical sense defined in Section~\ref{sec_pultr_functors}). 


\begin{thm}
\label{thm_adjunction_pultr_1}
Let $\mathfrak{T}$ be a connected $(\rho,\tau)$-Pultr template, let $\X$ and $\Y$ be a $\tau$-structure and a $\rho$-structure, respectively, let $k\in\N$, let $k'=(k+1)\cdot\diam(\mathfrak{T})$, and suppose that $\Lambda\X\leadsto^{k'}\Y$. Then $\X\leadsto^{k}\Gamma\Y$.
\end{thm}

The next result gives the quantum version of the other direction of Theorem~\ref{thm_pultr_adjoints}, provided that the Pultr template satisfies a technical property that we call faithfulness.

\begin{thm}
\label{thm_quantum_adjunction_faithful}
    Let $\mathfrak{T}$ be a faithful $(\rho,\tau)$-Pultr template, let $\X$ and $\Y$ be a $\tau$-structure and a $\rho$-structure, respectively, let $k\in\N$, and suppose that $\X\leadsto^{k}\Gamma\Y$. Then $\Lambda\X\leadsto^{k}\Y$.
\end{thm}

We point out that the two directions of Pultr adjunction behave differently with respect to local compatibility of the quantum strategies. Indeed, while the implication $\X\leadsto\Gamma\Y$ $\Rightarrow$ $\Lambda\X\leadsto\Y$ preserves the level of compatibility, the converse implication decreases it by a multiplicative factor depending on the structure of the Pultr template (in particular, its diameter, see formal definitions in Section~\ref{sec_pultr_functors}). This ``compatibility decay'' appears to be intrinsic in the way quantum strategies for the $\Lambda\X,\Y$ game are translated into quantum strategies for the $\X,\Gamma\Y$ game (see the proof of Theorem~\ref{thm_adjunction_pultr_1}), and it is the reason why arbitrarily high compatibility for $d$-to-$1$ instances is required in Conjecture~\ref{conj_d_to_1_khot_quantum}.
More in general, many classic reductions from hardness of approximation theory involve ``increasing the scope of constraints by capturing relations within chains of adjacent variables''---which is the intuitive description of the action of a central Pultr functor $\Gamma$. Examples of this are the reduction from $\textsc{Unique Games}$ to $\textsc{MaxCut}$ of~\cite{khot2007optimal} used in~\cite{mousavi2025quantum}; the reduction from weighted $d$-to-$1$ games to unweighted $d$-to-$d$ games from~\cite{Dinur09:sicomp}, see also~\cite{khot2008vertex} and Appendix~\ref{app_dinur_reductions} of the current paper; and the line-digraph reduction mentioned above.
Hence, while the current work focusses specifically on the colouring game, some level of compatibility of the quantum strategies for $d$-to-$1$ instances appears to be necessary in order to transfer pseudo-telepathy to other types of CSPs using general reductions from CSP approximation.

\subsection{The quantum chromatic gap via $d$-to-$1$ games}
\label{subsec_reductions}
The first main result of this paper is the following.

\begin{shaded}
\vspace{-.4cm}
\begin{thm}
\label{thm_main}
Suppose that Conjecture~\ref{conj_d_to_1_khot_quantum} holds for some $d\geq 2$. Then, for any $c\in\N$, there exists a graph with quantum chromatic number $3$ and classical chromatic number larger than $c$.
\end{thm}
\vspace{-.4cm}
\end{shaded}
To establish Theorem~\ref{thm_main},
we transfer pseudo-telepathy from the $d$-to-$1$ games to graph colouring via
the combination of the following transformations:
\begin{enumerate}
    \item a reduction from $d$-to-$1$ games to $d$-to-$d$ games from ~\cite{Dinur09:sicomp} (which is a modification of a reduction from~\cite{khot2008vertex} used to prove optimal inapproximability of vertex cover conditional to the Unique Games Conjecture); 
    \item a reduction from $d$-to-$d$ games to $2d$ vs. $\mathcal{O}(1)$ graph colouring from~\cite{GS20:icalp} (which, in turn, generalises a reduction from~\cite{Dinur09:sicomp} obtained in the case $d=2$);
    \item the \textit{line-digraph} reduction, whose behaviour with respect to the chromatic number is well known~\cite{HarnerE72,poljak1981arc}, and which has been often utilised in the context of hardness results or lower bounds against relaxations of approximate graph colouring, see e.g.~\cite{KOWZ22,GS20:icalp,cz23stoc:ba,HechtMS23}.
\end{enumerate}

In order for this process to transform pseudo-telepathic $d$-to-$1$ instances into graphs exhibiting an unbounded gap between quantum and classical chromatic number, we need them to be \textit{classically sound} (i.e., highly unsatisfiable instances are turned into highly chromatic graphs) and \textit{quantum complete} (i.e., instances admitting a perfect quantum assignment are turned into graphs having low quantum chromatic number). The first part comes for free from the soundness analyses performed in the context of approximation hardness. Our task is to establish quantum completeness. For the reduction in $1.$, this is done by emulating the classical completeness analysis, see~Appendix~\ref{app_dinur_reductions}.
As for the reductions in $2.$ and $3.$, we show that both of them are captured by Pultr functors associated with suitable Pultr templates meeting the required connectivity and faithfulness hypotheses of Theorems~\ref{thm_adjunction_pultr_1} and~\ref{thm_quantum_adjunction_faithful}. Hence, their quantum completeness is a consequence of the quantum Pultr adjunction. Full details of this part of the proof are given in Section~\ref{sec_reductions}.

\subsection{The quantum  chromatic gap via $3$XOR games}
\label{subsec_quantum_chromatic_gap_3xor}
How to construct strongly pseudo-telepathic $d$-to-$1$ games required for proving unbounded quantum chromatic gap?
A natural starting point is to look at non-local games encoding Boolean linear equations. Indeed, many of the well-known examples of violations of Bell's inequality involve cooperating entangled players proving that certain classically unsatisfiable systems of linear equations admit a solution---e.g., the  Greenberger--Horne--Zeilinger (GHZ) game~\cite{greenberger1989going,greenberger1990bell} and the Mermin--Peres magic square game~\cite{mermin1990simple,mermin1993hidden,peres1990incompatible}. Such games have a natural \textit{many-to-$1$} structure, in that, for every partial assignment to some equation, there exists
precisely one consistent assignment to each of the variables appearing in it. As a consequence, a first approach could be to consider a parallel repetition of the non-local game associated with Boolean linear equations as in~\cite{cleve2008perfect}, in order to decrease the classical value while preserving the existence of a perfect quantum assignment. The reason why this approach does not work is that the resulting game is a $d$-to-$1$ instance for increasingly large $d$ (in particular, $d$ would be a function of the soundness parameter $\epsilon$ in Conjecture~\ref{conj_d_to_1_khot_quantum}). However, transferring pseudo-telepathy to the colouring game in order to  obtain an unbounded chromatic gap requires $d$ to be fixed.

For this reason, we consider a suitably modified repetition of \textsc{3XOR} games (encoding Boolean linear equations having three variables per equation), where the players' answers are required to be consistent across the iterations of the game. 
By assuming pseudo-telepathy of such games (with any classical soundness threshold strictly smaller than $1$), we derive a $4$ vs. $\mathcal{O}(1)$ gap between quantum and classical chromatic numbers.
In the statement below, $\mathcal{G}(S,n)$ indicates the $n$-fold repetition of the game corresponding to a regular system $S$ of equations (see Section~\ref{sec_quantum_chromatic_gap_3xor} for the formal definitions).

\begin{shaded}
\vspace{-.4cm}
\begin{thm}
\label{thm_equations_to_colouring}
    Suppose that there exists some $0<s^*<1$ such that, for each $n\in\N$, there exists a regular instance $S$ of $3$XOR for which $\sat(S)<s^*$ but $\mathcal{G}(S,n)$ admits a perfect $1$-compatible quantum assignment. Then, for each $c\in\N$, there exists a graph with quantum chromatic number at most $4$ and classical chromatic number larger than $c$.
\end{thm}
\vspace{-.4cm}
\end{shaded}

Note that the assumption of Theorem~\ref{thm_equations_to_colouring} is true for $n=1$, because of the Mermin--Peres construction. 
The result is proved by showing that pseudo-telepathic instances for the $n$-fold repetition of $3$XOR can be transformed into \textit{strongly} pseudo-telepathic $2$-to-$2$ instances (i.e., with classical soundness arbitrarily close to $0$), via the geometric framework developed in~\cite{Khot17:stoc-independent,Khot18:focs-pseudorandom,Dinur18:stoc-non-optimality,Dinur18:stoc-towards} to prove hardness of approximating $2$-to-$2$ games with imperfect completeness. Then, the machinery developed in the previous part of the paper allows obtaining the required quantum chromatic gap.
The reduction we use to prove Theorem~\ref{thm_equations_to_colouring} skips the step $1.$ in Section~\ref{subsec_reductions}
and transforms the input $3$XOR game directly into a $2$-to-$2$ instance. Nevertheless, the case $d=2$, $k=1$ of Conjecture~\ref{conj_d_to_1_khot_quantum} can be derived from the hypothesis of Theorem~\ref{thm_equations_to_colouring} by applying the same proof to a modified, $2$-to-$1$ version of the \textsc{3XOR} $\rightarrow$ $2$-to-$2$ reduction, as described in~\cite{Dinur18:stoc-towards} (see Remark~\ref{rem_proving_conj}).
We also point out that, unlike Theorem~\ref{thm_main}, Theorem~\ref{thm_equations_to_colouring} does not allow improving the quantum chromatic gap from $4$ vs. $\mathcal{O}(1)$ to $3$ vs. $\mathcal{O}(1)$. The reason is that requiring local compatibility in a perfect quantum assignment of the $\mathcal{G}(S,n)$ game for any parameter larger than $1$ would collapse the assignment to a classical one, as it would enforce commutativity between \textit{all} pairs of measurements appearing in the strategy. As a consequence, this time, we cannot apply the line-digraph operator, since quantum completeness of the latter reduction requires higher compatibility, cf.~Theorem~\ref{thm_adjunction_pultr_1}.

Also in this case, the classical soundness comes essentially for free from the approximation hardness setting, where it is established by analysing the expansion properties of the Grassmann graphs encoding the geometry of the linear spaces constituting the blocks of the reduction, see Section~\ref{subsec_description_rho_one}.
Some care is only required when manipulating $2$-to-$2$ instances having small classical value, which are the output of the reduction, as a slightly different version of classical value (denoted by $\isat$, see Section~\ref{subsec_preprocessing}) is needed in order for the reduction from~\cite{GS20:icalp} to go through. We deal with this issue through a result by~\cite{DawarM25}. (Alternatively, the obstacle can be avoided by considering the $2$-to-$1$ version of the reduction---see the discussion in Section~\ref{subsec_classical_soundness_rho}.)  

The bulk of our work in this setting is to prove quantum completeness. Intuitively, the idea is the following. Given a \textsc{3XOR} instance $S$, the reduced $2$-to-$2$ instance has vertices corresponding to direct sums $L\oplus H_\bu$ of $(i)$ a linear space $H_\bu$ (over the Boolean field $\F_2$) encoding a list of equations $\bu$, and $(ii)$ a linear space $L$ of some dimension $\ell$ (which is a fixed parameter of the reduction) that is linearly independent with $H_\bu$. Moreover, the alphabet of the $2$-to-$2$ instance consists of linear functionals $\psi:L\oplus H_\bu\to\F_2$ that respect the equations in $\bu$. Suppose now that $S$ admits a perfect quantum strategy given by measurements $\{Q_{\bu,\vartheta}\}$, where $\bu$ is a list of equations appearing in the $n$-fold repetition of the $S$-game and $\vartheta$ encodes the players' answers. Then the measurements $\{W_{L\oplus H_\bu,\psi}\}$ in a quantum strategy for the reduced instance are obtained as the sum of all projectors $Q_{\bu,\vartheta}$ for which the linear functional corresponding to $\vartheta$ agrees with $\psi$ when restricted to the space $L\oplus H_\bu$. 
Using the consistency of the players' answers across iterations of the game and geometric properties of the reductions (in particular, the ``extension'' Lemma~\ref{lem_extension_side_conditions} from~\cite{Khot17:stoc-independent}), we show that the resulting quantum strategy is perfect, thus achieving quantum completeness.

To the best of our knowledge, this is the first usage of the reduction of~\cite{Khot17:stoc-independent,Khot18:focs-pseudorandom,Dinur18:stoc-non-optimality,Dinur18:stoc-towards} in the quantum setting. An interesting prospect for future work is to obtain a finer translation of quantum strategies for \textsc{3XOR} instances into quantum strategies for the reduced $2$-to-$2$ instances, relying on the rich geometric structure of the Grassman graphs encoding the latter. This would allow
weakening the assumptions in Theorem~\ref{thm_equations_to_colouring} and, thus, would make the task of exhibiting the required pseudo-telepathic \textsc{3XOR} instances easier. 
As it turns out, the reduction allows for a pseudo-telepathy transfer with \textit{perfect} completeness in the quantum setting (since \textsc{3XOR} is pseudo-telepathic), while the same does not hold in the classical approximation hardness setting (since \textsc{3XOR} is not NP-hard).
A similar phenomenon occurs in the context of integrality gaps for convex relaxations---specifically, the Sum-of-Squares semidefinite-programming hierarchy~\cite{Lasserre02,parrilo2000structured}---as observed in~\cite{Khot18:focs-pseudorandom} (see also~\cite[\S2.9]{khot2019proof}). Indeed, integrality gaps can be transferred from \textsc{3XOR} to $2$-to-$2$ games (and, then, to other CSPs, such as approximate graph colouring) with perfect completeness, thanks to the \textsc{3XOR} integrality gaps of~\cite{grigoriev2001linear,schoenebeck2008linear}.


\section{Classical and quantum CSPs}
\label{sec_classical_quantum_CSPs}
In this preliminary section, we formally describe classical and quantum assignments for constraint satisfaction problems. We adopt the language of
\textit{predicate satisfiability}, which yields the most natural description of specific types of CSPs that shall appear frequently (in particular, $d$-to-$1$ and $d$-to-$d$). Two alternative descriptions will sometimes be preferable: 
the language of \textit{non-local games}, providing a physical intuition of classical, quantum, and locally compatible quantum assignments as winning strategies for two-player cooperative games; and the language of \textit{homomorphisms between relational structures}, allowing for a more direct connection with categorical constructs such as the Pultr adjunction, as we shall see in Section~\ref{sec_pultr_functors}. 
The three, equivalent descriptions are given in Sections~\ref{subsec_CSP_predicates},~\ref{subsec_CSP_nonlocal_games}, and~\ref{subsec_CSP_homomorphisms}, respectively.

\subsection{Predicate satisfiability}
\label{subsec_CSP_predicates}

A CSP instance consists of a tuple $\Phi=(X,E,A,\pi,\phi)$ where 
\begin{itemize}
    \item $X$ is a finite set (whose elements are called \textit{variables});
    \item 
    $E$ is a multiset of \textit{constraint} tuples $\bx$ of variables, each having a finite length called \textit{arity} and denoted by $\ar(\bx)$;
    \item $A$ is a finite set (called the \textit{alphabet} or \textit{label set});
    \item $\pi$ is a positive probability distribution over $E$;
    \item $\phi$ is a collection of predicates $\phi_\bx:A^{\ar(\bx)}\to\{0,1\}$ for each $\bx\in E$.\footnote{We shall often identify the function $\phi_\bx$ with the set $\{\ba\in A^{\ar(\bx)}\suchThat\phi_\bx(\ba)=1\}$.}
\end{itemize}

A \textit{classical assignment} for a CSP instance $\Phi$ is a function $f:X\to A$.
The value of $f$ on a constraint $\bx$ is $\phi_\bx(f(\bx))$, where $f(\bx)$ is the entrywise application of $f$ to the entries of $\bx$.
The total value of $f$ is the weighted fraction of satisfied constraints; i.e., the quantity 
\begin{align*}
    \sat_f(\Phi)=\mathbb{E}_{\bx\sim\pi}\phi_\bx(f(\bx)).
\end{align*}
If $\sat_f(\Phi)=1$, we say that $f$ is a \textit{perfect} classical assignment.
The \textit{classical value} of $\Phi$, denoted by $\sat(\Phi)$, is the maximum of $\sat_f(\Phi)$ over all classical assignments $f:X\to A$.

Given a finite set $S$ and a finite-dimensional Hilbert space $\HH$, we let $\PVM_{\HH}(S)$ be the set of projection-valued measurements over $\HH$ whose outcomes are indexed by $S$. I.e., $\PVM_{\HH}(S)$ contains all sets $\{Q_s:s\in S\}$ 
where each $Q_s$ is a projector onto $\HH$ and, in addition, $\sum_{s\in S}Q_s$ is the identity projector $\id_\HH$.
A \textit{quantum assignment} over $\HH$ for a CSP instance $\Phi$ consists of a function $\mathfrak{Q}:X\to\PVM_\HH(A)$. For notational convenience, given a variable $x\in X$ and a label $a\in A$, we shall denote by $Q_{x,a}$ the $a$-th projector in the measurement $\mathfrak{Q}(x)$. The value of $\mathfrak{Q}$ on a constraint $\bx=(x_1,\dots,x_r)\in E$ of arity $r$ is the quantity
\begin{align*}
\frac{1}{\dim(\HH)}
    \sum_{\ba\in A^{r}}\phi_\bx(\ba)\,\tr(Q_{x_1,a_1}\cdot Q_{x_2,a_2}\cdot\ldots\cdot Q_{x_{r},a_{r}}).
\end{align*}
The total value of $\mathfrak{Q}$ is then the quantity
\begin{align}
\label{eqn_1527_232}
\qsat_{\mathfrak{Q}}(\Phi)=\frac{1}{\dim(\HH)}\mathbb{E}_{\bx\sim\pi}
    \sum_{\ba\in A^{\ar(\bx)}}\phi_\bx(\ba)\,\tr(Q_{x_1,a_1}\cdot Q_{x_2,a_2}\cdot \ldots\cdot Q_{x_{\ar(\bx)},a_{\ar(\bx)}}).
\end{align}
We say that $\mathfrak{Q}$ is a \textit{perfect} quantum assignment if $\qsat_{\mathfrak{Q}}(\Phi)=1$.
If all constraints of $\Phi$ have arity $1$ or $2$, the quantity in~\eqref{eqn_1527_232} is always real, and the assignment $\mathfrak{Q}$ is sometimes called a \textit{quantum noncommutative strategy} (see, for example,~\cite[\S~3.1]{mousavi2025quantum}).
%
%
In the general case, in order to enforce that the traces of products of projectors in the expression above be real, one further requires that projectors corresponding to variables appearing in a common constraint should commute. We call these quantum assignments \textit{$1$-compatible}. 
More in general, as discussed in Section~\ref{subsec_classical_quantum_CSPs}, we consider quantum assignments enforcing commutativity between projectors corresponding to variables that are close in the metric over $X$ induced by the constraints in $E$. 
Enforcing this type of \textit{local compatibility} for quantum assignments of a CSP instance is needed in order to 
achieve quantum completeness of Pultr reductions, which is the focus of Section~\ref{sec_pultr_functors}.
%

We now formally define locally compatible quantum assignments.
Given two variables $x,x'\in X$, we let $\dist_\Phi(x,x')$ denote the distance between $x$ and $x'$ in the \textit{Gaifman graph} of $\Phi$---i.e., the undirected graph having vertex set $X$, where two vertices are adjacent if and only if they appear in a common constraint. For two operators $Q$ and $Q'$, let $[Q,Q']=QQ'-Q'Q$ be their commutator. 
We say that a quantum assignment $\mathfrak{Q}$ is \textit{$k$-compatible} for some $k\in\N\cup\{0\}$ if $[Q_{x,a},Q_{x',a'}]=O$ for each pair $x,x'\in X$ such that $\dist_{\Phi}(x,x')\leq k$ and each pair $a,a'\in A$. Clearly, the $k$-compatibility requirement becomes progressively stronger as $k$ increases.
In the extreme case, if $\delta$ is the diameter of the Gaifman graph of $\Phi$, a $\delta$-compatible quantum assignment requires that the projectors in \textit{all} PVMs of the assignment should commute with each other and, thus, be simultaneously diagonalisable. Hence, $\delta$-compatible quantum assignments collapse to classical assignments.
Note that
projector commutativity is not transitive, so locally compatible quantum assignments do not automatically collapse to globally compatible, and thus classical, assignments.

The total value of a $k$-compatible quantum strategy is clearly real if $k\geq 1$, regardless of the arity of its constraints. Hence, the $k$-compatible quantum value of $\Phi$, denoted by $\qsat_k(\Phi)$, can be defined as the supremum of $\qsat_{\mathfrak{Q}}(\Phi)$ over all $k$-compatible quantum assignments $\mathfrak{Q}$.
The discussion above shows that the following chain of inequalities holds:
\begin{align*} 
    \qsat_1(\Phi)
    \geq
    \qsat_2(\Phi)
    \geq 
    \dots 
    \geq
    \qsat_\delta(\Phi)
    =
    \qsat_{\delta+1}(\Phi)
    =
    \dots
    =
    \sat(\Phi).
\end{align*}
In addition, in the unary or binary case, we have that $\qsat_0(\Phi)\geq \qsat_1(\Phi)$. Hence, the notion described above allows to smoothly decrease the value of a CSP from the quantum noncommutative to the classical one, by considering PVM assignments that ``look classical'' over progressively larger subinstances of $\Phi$. 
%

We now define a particular type of CSP instances that shall frequently appear in this work.
A \textit{label-cover} instance is a CSP instance $\Phi=(X,E,A,\pi,\phi)$ all of whose constraints are binary. Notice that, in this case, $(X,E)$ is a digraph (possibly with repeated edges). We say that $\Phi$ is \textit{bipartite} if
$E\subseteq X_1\times X_2$ for some decomposition $X=X_1\sqcup X_2$.
Following~\cite{dinur2014analytical,dinur2015parallel},
we say that a bipartite label-cover instance $\Phi$ is \textit{projective} if there is a decomposition $A=A_1\sqcup A_2$ such that, for each $\bx\in E$, it holds that $(i)$ $\phi_\bx\subseteq A_1\times A_2$, and $(ii)$ for each $a_1\in A_1$ there exists precisely one value $a_2\in A_2$ for which $\phi_\bx(a_1,a_2)=1$.\footnote{Sometimes, in the literature on both classical and quantum CSPs/non-local games, the expression ``label-cover instances'' indicates what we call here ``projective label-cover instances''.}
Clearly, if $f:X\to A$ is an optimal classical assignment for $\Phi$, there is an optimal classical assignment $f'$ such that $f'(X_1)\subseteq A_1$ and $f'(X_2)\subseteq A_2$. Hence, we can assume without loss of generality that optimal classical assignments satisfy the latter condition. (We shall see that the same holds for optimal quantum assignments, see Proposition~\ref{prop_only_relevant_projectors}.)
We say that a projective label-cover instance is \textit{$d$-to-$1$} for some $d\in\N$ if for each $\bx\in E$ and each $a_2\in A_2$ there exist precisely $d$-many values $a_1\in A_1$ for which $\phi_\bx(a_1,a_2)=1$. 
Observe that, in this case, $|A_1|=d\cdot|A_2|$.

\subsection{Non-local games}
\label{subsec_CSP_nonlocal_games}
The following two-player cooperative game can be associated with a CSP instance $\Phi=(X,E,A,\pi,\phi)$ (see for example~\cite{abramsky2017quantum}): 
\begin{itemize}
    \item The verifier sends Alice a constraint $\bx\in E$ sampled according to the probability distribution $\pi$, and Alice responds with a tuple $\ba\in A^{\ar(\bx)}$;
    \item the verifier sends Bob a variable $x\in X$ sampled uniformly at random, and Bob responds with a label $a\in A$.
\end{itemize}
Alice and Bob win the game if, whenever $x=x_i$ for some $i\in[\ar(\bx)]$, it holds that $a=a_i$. 
It is straightforward to see that a perfect classical (deterministic or randomised) strategy for Alice and Bob exists precisely when $\sat(\Phi)=1$. A \textit{quantum} strategy $\mathfrak{E}$ consists of a finite-dimensional Hilbert space $\HH$, a state $\psi\in\HH\otimes\HH$, and two collections of measurements---one for Alice and one for Bob---indexed by the verifier's questions and having as possible outcomes the admitted players' answers.
More in detail, there is a measurement $\{E_{\bx,\ba}\}_{\ba\in A^{\ar(\bx)}}\in\PVM_\HH(A^{\ar(\bx)})$ for each $\bx\in E$, and a measurement $\{{E'_{x,a}}\}_{a\in A}\in\PVM_\HH(A)$ for each $x\in X$. Upon receiving the verifier's questions, Alice and Bob select the corresponding PVMs and perform the corresponding measurement on their part of the state $\psi$; then, their answers to the verifier are the obtained outcomes.
Hence, in order for the strategy to be perfect, the following conditions must hold:
\begin{itemize}
    \item $\psi^*({E_{\bx,\ba}}
    \otimes {E'_{x,a}})\psi=0$ whenever $x$ appears in $\bx$ with index $i$ and $a\neq a_i$;
    \item $\psi^*({E_{\bx,\ba}}\otimes \id_\HH)\psi=0$ whenever $\phi_\bx(\ba)=0$.
\end{itemize}
The next result shows the connection between perfect quantum strategies and perfect quantum assignments.

\begin{thm}[\cite{abramsky2017quantum}]
\label{thm_abramsky_quantum_strategies_quantum_assignments}
    Let $\Phi$ be a $\CSP$ instance. Then 
    $\Phi$ has a perfect $1$-compatible quantum assignment if and only if the non-local game associated with $\Phi$ admits a perfect quantum strategy.
\end{thm}

We can easily modify the non-local game above to capture locally compatible quantum assignments. To that end, we add ``fake'' constraints to $\Phi$ that have no effect on classical strategies, but enforce commutativity of the projectors corresponding to the selected variables. Formally, we consider the CSP instance $\Phi^{(k)}=(X,E',A,\pi',\phi')$ where $E'$ is obtained from $E$ by adding a binary constraint $(x,x')$ for each two variables $x,x'$ such that $\dist_\Phi(x,x')\leq k$. The corresponding set of admitted labels is $\phi_{(x,x')}=A^2$; i.e., no restriction is posed on the value of the constraint. Finally, $\pi'$ assigns probability, say, $\frac{\pi(\bx)}{2}$ for each $\bx\in E$, and probability $\frac{1}{2\alpha}$ for each new binary constraint $(x,x')$, where $\alpha$ is the number of the new binary constraints. The following fact is immediate from Theorem~\ref{thm_abramsky_quantum_strategies_quantum_assignments}.

\begin{cor}
\label{cor_1239_2103}
Let $\Phi$ be a CSP instance, and take $k\in\N$. Then 
    $\Phi$ has a perfect $k$-compatible quantum assignment if and only if the non-local game associated with $\Phi^{(k)}$ admits a perfect quantum strategy.
\end{cor}

\subsection{Homomorphism problems}
\label{subsec_CSP_homomorphisms}
We end this preliminary section by giving an equivalent description of CSPs based on homomorphisms between relational structures.
A \emph{signature} $\rho$ is a finite set of relation symbols $R$, each with its \emph{arity} $\ar(R)\in\N$.
A \emph{relational structure} $\mathbf{S}$ with signature $\rho$ (in short, a $\rho$-structure) consists of a finite set $S$ and a relation $R({\mathbf{S}})\subseteq S^{\ar(R)}$ for each symbol $R\in\rho$. A \textit{homomorphism} $f:\mathbf{S}\to\T$ between two relational structures $\mathbf{S}$ and $\T$ on a common signature $\rho$ is a map from the domain $S$ of $\mathbf{S}$ to the domain $T$ of $\T$ that preserves all relations; i.e., $f(\bs)\in R({\T})$ for each symbol $R\in\rho$ and each tuple $\bs\in R({\mathbf{S}})$ (where $f(\bs)$ is the entrywise application of $f$ the the entries of $\bs$). We denote the existence of a homomorphism from $\mathbf{S}$ to $\T$ by the notation $\mathbf{S}\to\T$.
Given a CSP instance $\Phi=(X,E,A,\pi,\phi)$, we can associate with $\Phi$ two relational structures $\X$, $\A$ over a common signature $\rho$ as follows. 
The domains of $\X$ and of $\A$ are the variable set $X$ and the label set $A$, respectively. The signature $\rho$ contains an $r$-ary symbol for each distinct map $\phi_\bx:A^r\to\{0,1\}$; 
the interpretation $R(\X)$ of $R$ in $\X$ contains all tuples $\bx\in E$ such that $\phi_\bx$ corresponds to the symbol $R$, while
the interpretation of $R$ in $\A$ is the set $R(\A)=\{\ba\in A^r:\phi_\bx(\ba)=1\}$. 
Conversely, given two structures $\X$ and $\A$ on a common signature $\rho$, we can associate with them a CSP instance $\Phi=(X,E,A,\pi,\phi)$ such that $X$ and $A$ are the domains of $\X$ and $\A$, respectively; $E=\bigcup_{R\in\rho}R(\X)$; for $R\in\rho$ and $\bx\in R(\X)$, $\phi_\bx$ is the Boolean map defined by $\ba\mapsto 1$ for each $\ba\in R(\A)$ and $\ba\mapsto 0$ for each $\ba\in A^{\ar(R)}\setminus R(\A)$; and $\pi$ is any positive distribution over $E$. Note that, as far as \textit{perfect} (classical or quantum) assignments are concerned, the particular choice of $\pi$ is irrelevant.
Let $\X$ and $\A$ be the $\rho$-structures corresponding to a CSP instance $\Phi$. Clearly, it holds that $\sat(\Phi)=1$ if and only if $\X\to\A$. For $k\in\N$, we write $\X\leadsto^k\A$ if $\Phi$ admits a perfect $k$-compatible quantum assignment. (If $\X$ and $\A$ have only unary and binary relations, we also use the notation $\X\leadsto^0\A$ to indicate that $\Phi$ admits a perfect quantum assignment, with no local compatibility requirements.)
When $\Phi$ is the CSP instance corresponding to two structures $\X$ and $\A$, we denote $\dist_\Phi$ by $\dist_\X$.
\begin{example}
Let $\GG$ be an undirected graph, and view it as the relational structure having a single, binary relation containing, for each undirected edge $\{g,g'\}$ of $\GG$, the two tuples $(g,g')$ and $(g',g)$. Also, for $n\geq 2$, let $\K_n$ denote the $n$-clique.
Recall that the\textit{ classical chromatic number} of $\GG$ is the minimum $n$ such that $\GG\to\K_n$. The \textit{quantum chromatic number} of $\GG$ is defined as the minimum $n$ such that $\GG\leadsto^0\K_n$~\cite{CameronMNSW07}. 
\end{example}

The next proposition easily follows from a result proved in~\cite{abramsky2017quantum}.
%


\begin{prop}
\label{thm_quantum_assignments_as_projectors}
    Let $\X$ and $\Y$ be $\rho$-structures and let $k\in\N$. Then $\X\leadsto^k\Y$ if and only if
    there exists a finite-dimensional Hilbert space $\HH$ and
    a family of projectors $\{Q_{x,y}\}_{y\in Y}\in\PVM_\HH(Y)$ for each $x\in X$ such that
    \begin{enumerate}
    \item $Q_{x_1,y_1}\cdot Q_{x_2,y_2}\cdot\ldots\cdot Q_{x_{\ar(R)},y_{\ar(R)}}=O$ for each $R\in\rho$, $\bx\in R(\X)$, and $\by\in Y^{\ar(R)}\setminus R(\Y)$;
    \item $[Q_{x,y},Q_{x',y'}]=O$ for each $x,x'\in X$, $y,y'\in Y$ such that $\dist_\X(x,x')\leq k$.
\end{enumerate}
\end{prop}
\begin{proof}
    If $k=1$, the result is~\cite[Theorem~7]{abramsky2017quantum}. Suppose that $k\geq 2$, and let $\Phi$ be the CSP instance associated with $\X,\Y$. Let also $\Phi^{(k)}$ be the CSP instance obtained from $\Phi$ by adding binary constraints over variables at distance at most $k$ (as described in Section~\ref{subsec_CSP_nonlocal_games}), and let $\X^{(k)},\Y^{(k)}$ be the structures corresponding to $\Phi^{(k)}$. By definition, $\X\leadsto^k\Y$ is equivalent to $\Phi$ admitting a perfect $k$-compatible quantum assignment. By Corollary~\ref{cor_1239_2103}, the latter is equivalent to $\Phi^{(k)}$ admitting a perfect $1$-compatible quantum assignment---i.e., to $\X^{(k)}\leadsto^1\Y^{(k)}$. The result then follows by applying the proposition to the structures $\X^{(k)},\Y^{(k)}$, and noting that $\dist_{\X^{(k)}}(x,x')\leq 1$ $\Leftrightarrow$ $\dist_\X(x,x')\leq k$ for each $x,x'\in X$.
\end{proof}

\begin{rem}
We point out that Proposition~\ref{thm_quantum_assignments_as_projectors} also holds for $k=0$ if $\X$ and $\Y$ are graphs, as proved in~\cite[Corollary~2.2]{MancinskaR16}. In such case, a perfect quantum assignment corresponds to a perfect quantum strategy for a \textit{symmetric} version of the non-local game described in Section~\ref{subsec_CSP_nonlocal_games}, where both players receive a vertex as the verifier's question (such as the colouring game discussed in the Introduction). In particular, this means that the quantum chromatic number of a graph can be defined through the symmetric version of the game; for its formal description,
we refer the reader to~\cite{CameronMNSW07,MancinskaR16}.
\end{rem}

As a consequence,
we shall encode a perfect $k$-compatible quantum assignment for $\X$ and $\Y$ as a collection $\mathfrak{Q}=\{Q_{x,y}:x\in X,y\in Y\}$ of projectors onto some Hilbert space $\HH$ meeting the conditions of Proposition~\ref{thm_quantum_assignments_as_projectors}.
We will often make use of the following result.

\begin{prop}
\label{prop_compose_classic_quantum_homos}
    Let $\X,\X',\Y,\Y'$ be $\rho$-structures such that $\X\to\X'\leadsto^k\Y'\to\Y$ for some $k\in\N$. Then $\X\leadsto^k\Y$.
\end{prop}
\begin{proof}
    Let $f:\X\to\X'$ and $g:\Y'\to\Y$ be homomorphisms, and let $\mathfrak{Q}$ be a perfect $k$-compatible quantum assignment for $\X',\Y'$ over some Hilbert space $\HH$. Consider, for each $x\in X$ and $y\in Y$, the linear application
    \begin{align}
    \label{eqn_1244_1912}
        W_{x,y}
        =
        \sum_{y'\in g^{-1}(y)}Q_{f(x),y'}.
    \end{align}
    We claim that the set $\mathfrak{W}=\{W_{x,y}\suchThat x\in X,y\in Y\}$ yields a perfect $k$-compatible quantum assignment for $\X,\Y$ (over the same Hilbert space $\HH$).
First of all, the summands in~\eqref{eqn_1244_1912} are mutually orthogonal projectors and, thus, the linear applications in $\mathfrak{W}$ are projectors. The fact $\sum_{y\in Y}W_{x,y}=\id_\HH$ follows from the analogous property of $\mathfrak{Q}$. Hence, $\mathfrak{W}$ is a quantum assignment. 
Take now an $r$-ary symbol $R\in\rho$ and two tuples $\bx\in R(\X)$ and $\by\in Y^r\setminus R(\Y)$. Suppose, for the sake of contradiction, that
\begin{align*}
    O
    \neq
    \prod_{i\in [r]}W_{x_i,y_i}
    =
    \prod_{i\in [r]}\,\sum_{y'\in g^{-1}(y_i)}Q_{f(x_i),y'}.
\end{align*}
We deduce that there exist elements $y'_i\in g^{-1}(y_i)$ for each $i\in [r]$ such that
\begin{align}
\label{eqn_1912_1305}
    O\neq \prod_{i\in [r]}Q_{f(x_i),y'_i}.
\end{align}
Consider now the tuple $\by'=(y'_1,\dots,y'_r)$. Since $f$ is a homomorphism and $\bx\in R(\X)$, it holds that $f(\bx)\in R(\X')$. Since $\mathfrak{Q}$ is a perfect quantum assignment, we deduce from~\eqref{eqn_1912_1305} that $\by'\in R(\Y')$. Since $g$ is a homomorphism, we conclude that $\by=g(\by')\in R(\Y)$, a contradiction. This means that the assignment $\mathfrak{W}$ is perfect.
Finally, to show that $\mathfrak{W}$ is $k$-compatible, observe that $\dist_{\X'}(f(x),f(\tilde x))\leq \dist_\X(x,\tilde x)$ for each $x,\tilde x\in X$ since $f$ is a homomorphism. Hence, the $k$-compatibility of $\mathfrak{Q}$ implies that
\begin{align*}
    \Big[W_{x,y},W_{\tilde x,\tilde y}\Big]
    =
    \Big[\sum_{y'\in g^{-1}(y)}Q_{f(x),y'},\sum_{\tilde y'\in g^{-1}(\tilde y)}Q_{f(\tilde x),\tilde y'}\Big]
    =
    \sum_{y'\in g^{-1}(y)}\,\sum_{\tilde y'\in g^{-1}(\tilde y)}\Big[Q_{f(x),y'},Q_{f(\tilde x),\tilde y'}\Big]
    =
    O
\end{align*}
for each $x,\tilde x\in X$ with $\dist_\X(x,\tilde x)\leq k$ and each $y,\tilde y\in Y$.
Therefore, the assignment $\mathfrak{W}$ witnesses that $\X\leadsto^k\Y$, as claimed.
\end{proof}
\begin{rem}
    Following~\cite{ciardo_quantum_minion}, we might alternatively prove Proposition~\ref{prop_compose_classic_quantum_homos} by viewing a quantum assignment $\X\leadsto^k\Y$ as a classical assignment $\X^{(k)}\to\freeQ(\Y^{(k)})$, where $\X^{(k)}$ and $\Y^{(k)}$ are as in the proof of Proposition~\ref{thm_quantum_assignments_as_projectors} and $\freeQ$ denotes the \textit{free structure} generated by the \textit{quantum minion} $\Qminion$. (For the definitions, we refer the reader to~\cite{ciardo_quantum_minion}.) The proposition, then, easily follows by using the transitivity of classical homomorphisms. 
    We also point out that Proposition~\ref{prop_compose_classic_quantum_homos} cannot be strengthened to yield transitivity of perfect locally compatible quantum assignments, in the following sense.
    Suppose that $\X\leadsto^k\Y\leadsto^{k}\W$ is witnessed by two perfect $k$-compatible quantum assignments $\mathfrak{Q}=\{Q_{x,y}\}$ for $\X,\Y$ and $\mathfrak{Q}'=\{Q'_{y,w}\}$ for $\Y,\W$. Following~\cite{abramsky2017quantum}, one might define the perfect quantum assignment $\mathfrak{W}=\{W_{x,w}\}$ for $\X,\W$ by setting
    \begin{align*}
        W_{x,w}=\sum_{y\in Y}Q_{x,y}\otimes Q'_{y,w}.
    \end{align*}
    However, the $k$-compatibility of $\mathfrak{Q}$ and $\mathfrak{Q}'$ is not preserved in $\mathfrak{W}$, as the projectors in the latter assignment depend on \textit{all} projectors $Q'_{y,w}$ in $\mathfrak{Q}'$, rather than only those corresponding to vertices $y,y'\in Y$ having small distance.
\end{rem}

\section{Quantum adjunction of Pultr functors}
\label{sec_pultr_functors}
The goal of this section is to prove a quantum version of Theorem~\ref{thm_pultr_adjoints} on the adjunction of Pultr functors, which shall be needed in later sections to transfer pseudo-telepathy to the colouring game. We start by formally defining Pultr functors and Pultr templates.

\begin{defn}
\label{defn_pultr_template}
Let $\rho,\tau$ be two relational signatures. A $(\rho,\tau)$-\textit{Pultr template} $\mathfrak{T}$ consists of the following data: 
\begin{itemize}
    \item a $\rho$-structure $\A$;
    \item a $\rho$-structure $\B_{T}$ for each $T\in\tau$;
    \item a homomorphism $\epsilon_{i,T}:\A\to\B_{T}$ for each $T\in\tau$ and each $i\in [\ar(T)]$.
\end{itemize}
\end{defn}

We denote by $\mathbf{Str}(\rho)$ the set of all relational structures on signature $\rho$.

\begin{defn}
\label{defn_central_pultr_functor}
    Let $\mathfrak{T}$ be a $(\rho,\tau)$-Pultr template. The \textit{central Pultr functor} corresponding to $\mathfrak{T}$ is the map $\Gamma:\mathbf{Str}(\rho)\to\mathbf{Str}(\tau)$ defined as follows. For a $\rho$-structure $\X$, the domain of $\Gamma\X$ is the set of homomorphisms $h:\A\to\X$. For a symbol $T\in\tau$ of arity $t$, $T(\Gamma\X)$ contains all tuples $(\ell\circ\epsilon_{1,T},\dots,\ell\circ\epsilon_{t,T})$ for any homomorphism $\ell:\B_{T}\to\X$. 
\end{defn}

\begin{defn}
\label{defn_left_pultr_functor}
    Let $\mathfrak{T}$ be a $(\rho,\tau)$-Pultr template. The \textit{left Pultr functor} corresponding to $\mathfrak{T}$ is the map $\Lambda:\mathbf{Str}(\tau)\to\mathbf{Str}(\rho)$ defined as follows. For a $\tau$-structure $\X$, consider the $\rho$-structure $\mathbf{Z}$ given as the disjoint union of a copy $\A^{(x)}$ of $\A$ for each $x\in X$, and a copy $\B_T^{(\bx)}$ of $\B_T$ for each $T\in\tau$ and each $\bx\in T(\X)$. 
    Let $\sim$ be the coarsest equivalence relation over the vertices of $\mathbf{Z}$ that identifies the vertex $a^{(x_j)}$ of $\A^{(x_j)}$
    with the vertex $(\epsilon_{j,T}(a))^{(\bx)}$ of $\B_T^{(\bx)}$ for each $T\in\tau$, $\bx\in T(\X)$, $j\in [\ar(T)]$, and $a\in A$. We define $\Lambda\X$ as the quotient structure $\mathbf{Z}/\sim$.
\end{defn}
While the vertices of $\Lambda\X$ are formally equivalence classes of vertices of $\mathbf{Z}$, we will omit the equivalence class symbol for the sake of readability. For example, the expression ``the vertex $a^{(x)}$ of $\Lambda\X$'' should be understood as ``the equivalence class under $\sim$ of the vertex $a^{(x)}$ of $\mathbf{Z}$''.
%
%
A relational structure $\X$ is connected if its Gaifman graph is connected. In that case, the diameter of $\X$ (in symbols, $\diam(\X)$) is the diameter of the Gaifman graph.
We say that a $(\rho,\tau)$-Pultr template $\mathfrak{T}$ is \textit{connected} if the following two conditions hold:
\begin{itemize}
    \item Each of the $\rho$-structures $\A$, $\{\B_{T}\}_{T\in\tau}$ is connected;
    \item for each $R\in\rho$, $T\in\tau$, and $\bb\in R(\B_{T})$ there exists some $\ba\in R(\A)$ and some $i\in[\ar(T)]$ such that $\epsilon_{i,T}(\ba)=\bb$.
\end{itemize}
We define the diameter $\diam(\mathfrak{T})$ of a connected $(\rho,\tau)$-Pultr template $\mathfrak{T}$ as the maximum among the diameters of $\A$ and of $\B_{T}$ for each $T\in\tau$.

\begin{example}
\label{ex_line_digraph_pultr}
Let $\rho$ be the signature containing a unique, binary symbol $R$, and consider the $(\rho,\rho)$-Pultr template $\mathfrak{T}$ such that $\A$ is a directed path of length $1$ with edge $(a_1,a_2)$; $\B_R$ is a directed path of length $2$ with edges $(b_1,b_2)$ and $(b_2,b_3)$; the homomorphisms $\epsilon_{1,R}$ and $\epsilon_{2,R}:\A\to\B_R$ are defined as follows:
    $\epsilon_{1,R}$ maps $a_1$ to $b_1$ and $a_2$ to $b_2$; $\epsilon_{2,R}$ maps $a_1$ to $b_2$ and $a_2$ to $b_3$.
Then, the left and central Pultr functors corresponding to $\mathfrak{T}$ are the functors $\Lambda''$ and $\Gamma''$ from Example~\ref{ex_some_pultr_functors}. Observe that the Pultr template $\mathfrak{T}$ is connected.
\end{example}

\begin{thm*}
[Theorem~\ref{thm_adjunction_pultr_1} restated]
Let $\mathfrak{T}$ be a connected $(\rho,\tau)$-Pultr template, let $\X$ and $\Y$ be a $\tau$-structure and a $\rho$-structure, respectively, let $k\in\N$, let $k'=(k+1)\cdot\diam(\mathfrak{T})$, and suppose that $\Lambda\X\leadsto^{k'}\Y$. Then $\X\leadsto^{k}\Gamma\Y$.
\end{thm*}
\begin{proof}
    Let $\mathfrak{Q}=\{Q_{z,y}\suchThat z\in\Lambda\X,\,y\in Y\}$ be a perfect $k'$-compatible quantum assignment for the $\Lambda\X,\Y$ game. Given a vertex $x\in X$ and a map $h:A\to Y$ (not necessarily a homomorphism $\A\to\Y$), consider the linear application
    \begin{align}
    \label{eqn_1445_2911}
        W_{x,h}=\prod_{a\in A}Q_{a^{(x)},h(a)}.
    \end{align}
    Consider the set $\mathfrak{W}=\{W_{x,h}\suchThat x\in X, h:\A\to\Y\}$. We aim to show that $\mathfrak{W}$ is a quantum assignment witnessing that $\X\leadsto^k\Gamma\Y$.

    \begin{claim}
    \label{eqn_2911_1449}
     All linear applications in $\mathfrak{W}$ are projectors. 
    \end{claim}

    \begin{proof}
        For each $a,a'\in A$, $x\in X$, we have
        \begin{align*}
            \dist_{\Lambda\X}(a^{(x)},a'^{(x)})\leq \dist_{\A}(a,a')\leq\diam(\A)\leq\diam(\mathfrak{T})\leq k'.
        \end{align*}
        Since the strategy $\mathfrak{Q}$ is $k'$-compatible, we deduce that the projectors $Q_{a^{(x)},h(a)}$ and $Q_{a'^{(x)},h(a')}$ commute for each $h:\A\to \Y$. Hence, the product in~\eqref{eqn_1445_2911} is a projector.
    \end{proof}

    \begin{claim}
    \label{claim_2911_1529}
        If $h:A\to Y$ is not a homomorphism from $\A$ to $\Y$, $W_{x,h}=O$ for each $x\in X$.
    \end{claim}
    \begin{proof}
        Since $h$ is not a homomorphism, there exists some symbol $R\in\rho$ of arity $r$ and some tuple $\ba\in R(\A)$ for which $h(\ba)\not\in R(\Y)$. By definition of $\Lambda\X$, we have $\ba^{(x)}=(a_1^{(x)},\dots,a_r^{(x)})\in R(\Lambda\X)$. Since $\mathfrak{Q}$ is a perfect strategy, it follows that
        \begin{align*}
            \prod_{i\in [r]}Q_{a_i^{(x)},h(a_i)}=O.
        \end{align*}
        Let $\tilde A\subseteq A$ be the set of all elements of $A$ that do not appear in the tuple $\ba$. Using the $\diam(\A)$-compatibility of $\mathfrak{Q}$, we find
        \begin{align*}
            W_{x,h}=
            \prod_{a\in A}Q_{a^{(x)},h(a)}
            =
            \Big(\prod_{a\in \tilde A}Q_{a^{(x)},h(a)}\Big)\cdot\Big(\prod_{i\in [r]}Q_{a_i^{(x)},h(a_i)}\Big)
            =
            O,
        \end{align*}
        as claimed.
    \end{proof}

    \begin{claim}
    \label{claim_1559_0803}
        $\mathfrak{W}$ is a quantum assignment for the $\X,\Gamma\Y$ game.
    \end{claim}

    \begin{proof}
%
%
For each vertex $x\in X$, we have
\begin{align*}
    \sum_{h:\A\to\Y}W_{x,h}
    =
    \sum_{h:\A\to\Y}\prod_{a\in A}Q_{a^{(x)},h(a)}
    =
    \sum_{h:A\to Y}\prod_{a\in A}Q_{a^{(x)},h(a)}
    =
    \prod_{a\in A}\sum_{y\in Y}Q_{a^{(x)},y}
    =\prod_{a\in A}\id_\HH
    =\id_\HH,
\end{align*}
where the second equality holds because of Claim~\ref{claim_2911_1529}. It follows that $\mathfrak{W}$ is a quantum assignment, as claimed.
\end{proof}

\begin{claim}
\label{claim_1655_2911}
    The quantum assignment $\mathfrak{W}$ is perfect.
\end{claim}

\begin{proof}
    Take a symbol $T\in\tau$ of arity $t$, a tuple $\bx=(x_1,\dots,x_{t})\in T(\X)$, and a tuple $\bh=(h_1,\dots,h_{t})$ of homomorphisms $h_i:\A\to\Y$ such that $\bh\not\in T(\Gamma\Y)$.
 The claim would follow if we manage to show that 
\begin{align}
    \label{eqn_1652_2111}
    \prod_{i\in [t]}W_{x_i,h_i}=O.
\end{align}

Suppose, for the sake of contradiction, that~\eqref{eqn_1652_2111} is false.
Consider the $\rho$-structure $\B_{T}$. We aim to define a homomorphism $\ell:\B_T\to\Y$. 
Using the fact that the Pultr template $\mathfrak{T}$ is connected, we proceed as follows. Take $b\in B_T$. Since $\B_T$ is connected, there must exist some symbol $R\in\rho$ of arity $r$, some tuple $\bb\in R(\B_T)$, and some index $j\in [r]$ such that $b=b_j$. Using the second property of connected Pultr templates, we deduce that $\bb=\epsilon_{i,T}(\ba)$ for some $\ba\in R(\A)$ and some $i\in [t]$. Recalling that $h_i$ is a map from $A$ to $Y$, we define
\begin{align}
\label{eqn_1611_0803}
    \ell(b)=h_i(a_j).
\end{align}

\begin{subclaim}
    The map $\ell:B_T\to Y$ is well defined.
\end{subclaim}
\begin{proof}
    Suppose that there exist $\tilde R\in\rho$ of arity $\tilde r$, $\tilde\bb\in \tilde R(\B_T)$, $\tilde j\in [\tilde r]$, $\tilde\ba\in \tilde R(\A)$, and $\tilde i\in [t]$ such that $b=\tilde b_{\tilde j}$ and $\tilde\bb=\epsilon_{\tilde i,T}(\tilde\ba)$. Observe that
    \begin{align}
    \label{eqn_1604_2911}
        a_j^{(x_i)}
        =
        [\epsilon_{i,T}(a_j)]^{(\bx)}
        =
        b_j^{(\bx)}
        =
        b^{(\bx)}
        =
        \tilde b_{\tilde j}^{(\bx)}
        =
        [\epsilon_{\tilde i,T}(\tilde a_{\tilde j})]^{(\bx)}
        =
        \tilde a_{\tilde j}^{(x_{\tilde i})}.
    \end{align}
    Furthermore, given any two elements $\alpha,\alpha'\in A$ and any two indices $\iota,\iota'\in [t]$, we have
    \begin{align*}
        \dist_{\Lambda\X}(\alpha^{(x_\iota)},\alpha'^{(x_{\iota'})})
        =
        \dist_{\Lambda\X}([\epsilon_{\iota,T}(\alpha)]^{(\bx)},[\epsilon_{\iota',T}(\alpha')]^{(\bx)})
        \leq 
        \diam(\B_T)
        \leq \diam(\mathfrak{T})
        \leq k'.
    \end{align*}
    As a consequence, the terms in the product 
    \begin{align*}
        \prod_{i\in[t]}W_{x_i,h_i}=
        \prod_{i\in[t]}\prod_{a\in A}Q_{a^{(x_i)},h_i(a)}
    \end{align*}
    can be rearranged in a way that the two projectors $Q_{a_j^{(x_i)},h_i(a_j)}$ and $Q_{\tilde a_{\tilde j}^{(x_{\tilde i})},h_{\tilde i}(\tilde a_{\tilde j})}$ appear consecutively. As we are assuming that the product above is nonzero, we deduce from~\eqref{eqn_1604_2911} and from the fact that $\mathfrak{Q}$ is a perfect quantum assignment that $h_i(a_j)=h_{\tilde i}(\tilde a_{\tilde j})$. This means that the function $\ell:B_T\to Y$ is well defined, as claimed. 
\end{proof}

\begin{subclaim}
\label{claim_1925_2911}
    The map $\ell$ is a homomorphism from $\B_T$ to $\Y$.
\end{subclaim}
\begin{proof}
    Take a symbol $R\in\rho$ and a tuple $\bb\in R(\B_T)$. Using the connectivity of $\mathfrak{T}$, we find some $\ba\in R(\A)$ and some $i\in [t]$ such that $\bb=\epsilon_{i,T}(\ba)$.
    It then follows from the description of $\ell$ that $\ell(\bb)=h_i(\ba)$. Since 
    $h_i$ is a homomorphism from $\A$ to $\Y$, it preserves $R$. Therefore, we have that $h_i(\ba)\in R(\Y)$, thus proving the claim.
\end{proof}

\begin{subclaim}
\label{claim_1924_2911}
    The equality $\ell\circ\epsilon_{i,T}=h_i$ holds for each $i\in [t]$.
\end{subclaim}
\begin{proof}
    Take $a\in A$, and call $b=\epsilon_{i,T}(a)$. Since the $\rho$-structure $\A$ is connected, there exists $R\in\rho$ of arity $r$, $\ba\in R(\A)$, and $j\in[r]$ such that $a=a_j$. Define $\bb=\epsilon_{i,T}(\ba)$.
    Since $\epsilon_{i,T}$ is a homomorphism, $\bb\in R(\B_T)$. Note also that $b=b_j$. By definition of $\ell$, we have that $\ell(b)=h_i(a_j)$. Therefore,
    \begin{align*}
        h_i(a)
        =
        h_i(a_j)
        =
        \ell(b)
        =
        \ell(b_j)
        =
        \ell(\epsilon_{i,T}(a_j))
        =
        \ell\circ\epsilon_{i,T}(a),
    \end{align*}
    so the maps $\ell\circ\epsilon_{i,T}$ and $h_i$ coincide.
\end{proof}
Using Claim~\ref{claim_1924_2911}, we find
\begin{align*}
    \bh
    =
    (h_1,\dots,h_t)
    =
    (\ell\circ\epsilon_{1,T},\dots,\ell\circ\epsilon_{t,T}).
\end{align*}
Since, by Claim~\ref{claim_1925_2911}, $\ell$ is a homomorphism $\B_T\to\Y$, it follows from Definition~\ref{defn_central_pultr_functor} that $\bh\in T(\Gamma\Y)$, contradicting our assumption. This concludes the proof of Claim~\ref{claim_1655_2911}.
\end{proof}

\begin{claim}
    The quantum assignment $\mathfrak{W}$ is $k$-compatible.
\end{claim}

\begin{proof}
    Take two vertices $x,x'\in X$ such that $\dist_\X(x,x')=d\leq k$. We need to show that
    \begin{align}
        \label{1930_2911}
        \Big[W_{x,h},W_{x',h'}\Big]=O
    \end{align}
    for each $h,h':\A\to\Y$. The fact that $\dist_\X(x,x')=d$ means that there exist
    \begin{itemize}
        \item symbols $T_1,\dots,T_d\in\tau$;
        \item vertices $x_0=x,x_1,\dots,x_d=x'\in X$;
        \item tuples $\bx^{(i)}\in T_i(\X)$ for each $i\in [d]$;
        \item indices $\alpha_i,\omega_i\in[\ar(T_i)]$ for each $i\in[d]$
    \end{itemize}
    such that
        $x_{i-1}=x^{(i)}_{\alpha_i}$ and $x_i=x^{(i)}_{\omega_i}$
    for each $i\in [d]$. Take $a,a'\in A$. Using the triangle inequality, we obtain
    \begin{align}
    \label{eqn_1953_2911}
    \notag
        \dist_{\Lambda\X}(a^{(x)},a'^{(x')})
        &\leq
        \dist_{\Lambda\X}(a^{(x)},a^{(x')})+
        \dist_{\Lambda\X}(a^{(x')},a'^{(x')})\\
        &\leq \sum_{i\in [d]}\dist_{\Lambda\X}(a^{(x_{i-1})},a^{(x_i)})+
        \dist_{\Lambda\X}(a^{(x')},a'^{(x')}).
    \end{align}
    For each $i\in [d]$, we have
    \begin{align*}
        a^{(x_{i-1})}
        &=
        a^{(x^{(i)}_{\alpha_i})}
        =
        [\epsilon_{\alpha_i,T_i}(a)]^{\bx^{(i)}}\quad\quad\mbox{and}\quad\quad
        a^{(x_{i})}
        =
        a^{(x^{(i)}_{\omega_i})}
        =
        [\epsilon_{\omega_i,T_i}(a)]^{\bx^{(i)}}.
    \end{align*}
    Hence,

    \begin{align*}
        \dist_{\Lambda\X}(a^{(x_{i-1})},a^{(x_i)})
        &=
        \dist_{\Lambda\X}([\epsilon_{\alpha_i,T_i}(a)]^{\bx^{(i)}},[\epsilon_{\omega_i,T_i}(a)]^{\bx^{(i)}})
        \leq 
        \dist_{\B_T}(\epsilon_{\alpha_i,T_i}(a),\epsilon_{\omega_i,T_i}(a))\\
        &\leq
        \diam(\B_T)
        \leq
        \diam(\mathfrak{T}).
    \end{align*}
    Also, reasoning as in the proof of Claim~\ref{eqn_2911_1449}, we have
    \begin{align*}
        \dist_{\Lambda\X}(a^{(x')},a'^{(x')})
        \leq
        \dist_\A(a,a')
        \leq
        \diam(\A)
        \leq\diam(\mathfrak{T}).
    \end{align*}
    As a consequence,~\eqref{eqn_1953_2911} yields
    \begin{align*}
        \dist_{\Lambda\X}(a^{(x)},a'^{(x')})
        &\leq (d+1)\diam(\mathfrak{T})\leq (k+1)\diam(\mathfrak{T})=k'.
    \end{align*}
    We can then leverage the $k'$-compatibility of $\mathfrak{Q}$ to conclude that the projectors $Q_{a^{(x)},y}$ and $Q_{a'^{(x')},y'}$ commute for each $y,y'\in Y$ and, thus, 
\begin{align*}
    \Big[W_{x,h},W_{x',h'}\Big]
    =
    \Big[\prod_{a\in A}Q_{a^{(x)},h(a)},\prod_{a'\in A}Q_{a'^{(x')},h'(a')}\Big]
    =
    O,
\end{align*}
as claimed.
\end{proof}
    In summary, we have shown that $\mathfrak{W}$ yields a perfect $k$-compatible quantum assignment for the $\X,\Gamma\Y$ game, which means that $\X\leadsto^k\Gamma\Y$. 
\end{proof}


We now turn to prove the converse direction of the quantum Pultr adjunction.
We say that a $(\rho,\tau)$-Pultr template $\mathfrak{T}$ is \textit{faithful} if the following two properties hold for each $T\in\tau$:
\begin{itemize}
    \item $B_{T}=\bigsqcup_{i\in [\ar(T)]}\Imm(\epsilon_{i,T})$; 
    \item $\epsilon_{i,T}$ yields an isomorphism\footnote{I.e., a bijective homomorphism whose inverse map is a homomorphism.} from $\A$ to the substructure of $\B_T$ induced by $\Imm(\epsilon_{i,T})$ for each $i\in[\ar(T)]$.
\end{itemize}

\begin{example}
Fix a digraph $\GG$, and let $n$ be the number of its vertices. 
Let $\rho$ be the signature containing a unique, binary symbol $R$, and consider the $(\rho,\rho)$-Pultr template $\mathfrak{T}$ such that $\A$ consists of $n$ isolated vertices, $\B_R$ is the direct product $\GG\times\K_2$ (where $\K_2$ is the graph consisting of a single, undirected edge), and $\epsilon_{1,R},\epsilon_{2,R}$ are the two natural injections of $\A$ into $\B_R$.
As described in~\cite[\S2]{foniok2013adjoint}, the left and central Pultr functors corresponding to $\mathfrak{T}$ are the functors $\Lambda'$ and $\Gamma'$ from Example~\ref{ex_some_pultr_functors}. Observe that the Pultr template $\mathfrak{T}$ is faithful.
\end{example}

\begin{thm*}
[Theorem~\ref{thm_quantum_adjunction_faithful} restated]
    Let $\mathfrak{T}$ be a faithful $(\rho,\tau)$-Pultr template, let $\X$ and $\Y$ be a $\tau$-structure and a $\rho$-structure, respectively, let $k\in\N$, and suppose that $\X\leadsto^{k}\Gamma\Y$. Then $\Lambda\X\leadsto^{k}\Y$.
\end{thm*}
\begin{proof}
    Fix a perfect $k$-compatible quantum assignment $\mathfrak{Q}=\{Q_{x,f}:x\in X,f:\A\to\Y\}$ for the $\X,\Gamma\Y$ game. For $a\in A$, $x\in X$, and $y\in Y$, consider the linear application
    \begin{align}
    \label{eqn_1615_2211_a}
        W_{a^{(x)},y}
        =
        \sum_{\substack{f:\A\to\Y\\f(a)=y}}Q_{x,f}.
    \end{align}
    Also, for $T\in\tau$ of arity $t$, $\bx\in T(\X)$, $b\in B_{T}$, and $y\in Y$, consider the linear application
    \begin{align}
    \label{eqn_1615_2211_b}
        W_{b^{(\bx)},y}
        =
        \sum_{\substack{\ell:\B_{T}\to\Y\\\ell(b)=y}}\;\prod_{i\in [t]}Q_{x_i,\ell\circ\epsilon_{i,T}}.
    \end{align}
    We aim to show that the set 
    \begin{align}
    \label{eqn_1727_2211}
    \mathfrak{W}=\{W_{a^{(x)},y}\suchThat a\in A, x\in X, y\in Y\}\cup \{W_{b^{(\bx)},y}\suchThat  T\in\tau, \bx\in T(\X), b\in B_{T}, y\in Y\}
    \end{align}
    is a perfect $k$-compatible quantum assignment for the $\Lambda\X,\Y$ game.
    Before starting off with the analysis of $\mathfrak{W}$, it shall be useful to prove the next two facts.

    \begin{claim}
    \label{claim_1628_2211}
        Fix a symbol $T\in\tau$ of arity $t$, a tuple $\bx\in T(\X)$, and a function $\ell:B_{T}\to Y$. If 
        \begin{align}
            \label{16_09_2211}
            \prod_{i\in [t]}Q_{x_i,\ell\circ\epsilon_{i,T}}\neq O,
        \end{align}
        then $\ell$ is a homomorphism from $\B_{T}$ to $\Y$.
    \end{claim}
    \begin{proof}
    %
    Using that the quantum assignment $\mathfrak{Q}$ is $k$-compatible (and thus, in particular, $1$-compatible), from the fact that $\bx\in T(\X)$ we deduce that the tuple $(\ell\circ\epsilon_{1,T},\dots,\ell\circ\epsilon_{t,T})$ belongs to the relation $T(\Gamma\Y)$. Hence, there exists some homomorphism $\tilde\ell:\B_T\to\Y$ such that $\tilde\ell\circ\epsilon_{i,T}=\ell\circ\epsilon_{i,T}$ for each $i\in [t]$. We now use the fact that the Pultr Template $\mathfrak{T}$ is faithful---specifically, that the union of $\Imm(\epsilon_{i,T})$ over $i\in [t]$ is the whole set $B_T$. We deduce that $\ell=\tilde\ell$, which concludes the proof of the claim.
    \end{proof}

    \begin{claim}
    \label{claim_1248_2311}
        Fix a symbol $T\in\tau$ of arity $t$, a tuple $\bx\in T(\X)$, and two distinct homomorphisms $\ell,\ell':\B_{T}\to \Y$. It holds that
        \begin{align*}
        \Big(\prod_{i\in [t]}Q_{x_i,\ell\circ\epsilon_{i,T}}\Big)\cdot \Big(\prod_{i\in [t]}Q_{x_i,\ell'\circ\epsilon_{i,T}}\Big)=O.
    \end{align*}
    \end{claim}
    \begin{proof}
        Since $\mathfrak{T}$ is faithful, there exists some $j\in[t]$ such that $\ell\circ\epsilon_{j,T}\neq\ell'\circ\epsilon_{j,T}$, which means that $Q_{x_j,\ell\circ\epsilon_{j,T}}\cdot Q_{x_j,\ell'\circ\epsilon_{j,T}}=O$. The $1$-compatibility of $\mathfrak{Q}$ allows rearranging at wish the projectors $Q_{x_i,\ell\circ\epsilon_{i,T}}$ and $Q_{x_i,\ell'\circ\epsilon_{i,T}}$ appearing in the product on the left-hand side of~\eqref{claim_1248_2311} and, thus, the claim follows.
    \end{proof}

    \begin{claim}
    \label{claim_1308_2311}
        The linear applications defined in~\eqref{eqn_1615_2211_a} and~\eqref{eqn_1615_2211_b} respect the equivalence relation $\sim$ of Definition~\ref{defn_left_pultr_functor}.
    \end{claim}

\begin{proof}
    Recall that $\sim$ identifies vertices $v=a^{(x_j)}$ and $w=[\epsilon_{j}(a)]^{(\bx)}$ for each $T\in\tau$, $\bx\in T(\X)$, $j\in [\ar(T)]$, and $a\in A$. Thus, we need to show that the equality 
    \begin{align}
    \label{eqn_1723_2211}
        W_{v,y}
        =
        W_{w,y}
    \end{align}
    holds for $v$ and $w$ as above, for any $y\in Y$. Observe that
    \begin{align}
    \label{eqn_1704_2211}
        W_{w,y}
        =
        \sum_{\substack{\ell:\B_T\to\Y\\\ell(\epsilon_{j,T}(a))=y}}\;\prod_{i\in [t]}Q_{x_i,\ell\circ\epsilon_{i,T}}.
    \end{align}
    For each $i\in [t]$, let $\B_i$ be the substructure of $\B_T$ induced by $\Imm(\epsilon_{i,T})$; since $\mathfrak{T}$ is faithful, $\B_T$ is the disjoint union of the $\rho$-structures $\B_i$'s. 
    Consider the two sets
    \begin{align*}
        \mathcal{L}&=\{\ell:\B_T\to\Y\mbox{ s.t. } \ell(\epsilon_{j,T}(a))=y\},\\ 
        \mathcal{\hat{L}}&=\{(\ell_1,\dots,\ell_{t})\mbox{ s.t. }\ell_i:\B_i\to\Y\;\forall i\in [t],\; \ell_j(\epsilon_{j,T}(a))=y\}.
    \end{align*}
    By assigning the tuple $(\ell|_{B_1},\dots,\ell|_{B_{t}})$ to any $\ell\in\mathcal{L}$, we see that $\mathcal{L}\subseteq\mathcal{\hat L}$. The inclusion might be strict if not all constraints of $\B_T$ are inside some $\B_i$. Nevertheless, we know from Claim~\eqref{claim_1628_2211} that those functions $\ell$ that are not homomorphisms yield a trivial contribution to the sum in~\eqref{eqn_1704_2211}. Thus, we can rewrite the sum as
    \begin{align*}
        W_{w,y}
        =\sum_{(\ell_1,\dots,\ell_{t})\in\mathcal{\hat L}}\;\prod_{i\in [t]}Q_{x_i,\ell_i\circ\epsilon_{i,T}}.
    \end{align*}
Now we use that each $\epsilon_{i,T}$ yields an isomorphism from $\A$ to $\B_i$, which gives
\begin{align*}
    W_{w,y}
        &=
        \sum_{\substack{(f_1,\dots,f_{t})\\ f_i:\A\to\Y\;\forall i\in [t]\\f_j(a)=y}}\;\prod_{i\in [t]}Q_{x_i,f_i}
        =      \Big(\sum_{\substack{f:\A\to\Y\\f(a)=y}}Q_{x_j,f}
        \Big)\cdot\Big(\prod_{i\in [t]\setminus \{j\}}\;\sum_{f:\A\to\Y}Q_{x_i,f}\Big)\\
        &=      \Big(\sum_{\substack{f:\A\to\Y\\f(a)=y}}Q_{x_j,f}
        \Big)\cdot\Big(\prod_{i\in [t]\setminus \{j\}}\id_{\HH}\Big)
        =
        \sum_{\substack{f:\A\to\Y\\f(a)=y}}Q_{x_j,f}
        =
        W_{v,y},
\end{align*}
where the third equality follows from the fact that $\mathfrak{Q}$ is a quantum assignment for $\X,\Gamma\Y$. This establishes~\eqref{eqn_1723_2211} and concludes the proof of the claim.
\end{proof}

We now prove that
the set $\mathfrak{W}$ defined in~\eqref{eqn_1727_2211} yields a perfect quantum assignment for the $\Lambda\X,\Y$ game.

\begin{claim}
\label{eqn_1348_2311}
 All linear applications in $\mathfrak{W}$ are projectors. 
\end{claim}
\begin{proof}
    The fact that the linear application $W_{a^{(x)},y}$ defined in~\eqref{eqn_1615_2211_a} is a projector immediately follows by recalling that $Q_{x,f}\cdot Q_{x,f'}=O$ whenever $f\neq f':\A\to\Y$, since $\mathfrak{Q}$ is a perfect quantum assignment. Consider now a linear application $W_{b^{(\bx)},y}$ as defined in~\eqref{eqn_1615_2211_b}. Since $\mathfrak{Q}$ is a $1$-compatible assignment, the projectors $Q_{x_i,\ell\circ\epsilon_{i,T}}$ and $Q_{x_{i'},\ell\circ\epsilon_{i',T}}$ commute  for each $\ell:\B_T\to\Y$ and each $i,i'\in [t]$ and, thus, the product $\prod_{i\in [t]}Q_{x_i,\ell\circ\epsilon_{i,T}}$
    is a projector. Moreover, the product of any two projectors $\prod_{i\in [t]}Q_{x_i,\ell\circ\epsilon_{i,T}}$ and $\prod_{i\in [t]}Q_{x_i,\ell'\circ\epsilon_{i,T}}$ is zero for any distinct $\ell\neq\ell':\B_T\to\Y$, as shown in Claim~\ref{claim_1248_2311}. Hence, their sum $W_{b^{(\bx)},y}$ is a projector.
\end{proof}

\begin{claim} $\mathfrak{W}$ is a quantum assignment for the $\Lambda\X,\Y$ game.
\end{claim}
\begin{proof}
Using that $\mathfrak{Q}$ is a quantum assignment, it is straightforward to check from the definition~\eqref{eqn_1615_2211_a} that 
\begin{align*}
    \sum_{y\in Y}W_{a^{(x)},y}
    =
    \sum_{f:\A\to\Y}Q_{x,f}
    =
    \id_\HH
\end{align*}
for any $a\in A$ and $x\in X$.
Similarly, for $T\in\tau$ of arity $t$, $\bx\in T(\X)$, and $b\in B_{T}$, we have
\begin{align*}
    \sum_{y\in Y}W_{b^{(\bx)},y}
    =
    \sum_{\ell:\B_T\to\Y}\;\prod_{i\in [t]}Q_{x_i,\ell\circ\epsilon_{i,T}}
    =
    \prod_{i\in [t]}\sum_{f:\A\to\Y}Q_{x_i,f}
    =
    \prod_{i\in [t]}\id_\HH=\id_\HH,
\end{align*}
where we have used the same argument as in the proof of Claim~\ref{claim_1308_2311}.
\end{proof}

\begin{claim} The quantum assignment $\mathfrak{W}$ is perfect.
\end{claim}
\begin{proof}
Fix a symbol $R\in\rho$ of arity $r$ and a tuple $\by\in Y^r\setminus R(\Y)$. Recall from Definition~\ref{defn_left_pultr_functor} that the elements of $R(\Lambda\X)$ can be of two types, based on whether they correspond to $(i)$ elements of $R(\A)$, or $(ii)$ elements of $R(\B_{T})$ for some $T\in\tau$.
In the first case, we have $\ba^{(x)}=(a_1^{(x)},\dots,a_{r}^{(x)})\in R(\Lambda\X)$ for some $\ba\in R(\A)$ and $x\in X$. Observe that
\begin{align*}
    \prod_{j\in [r]}W_{a_j^{(x)},y_j}
    =
    \prod_{j\in [r]}\sum_{\substack{f:\A\to\Y\\f(a_j)=y_j}}Q_{x,f}.
\end{align*}
Distributing the product over the sums in the expression above, and recalling that $\{Q_{x,f}\}_{f:\A\to\Y}$ is a PVM and, thus, its members are mutually orthogonal projectors, we obtain
\begin{align*}
    \prod_{j\in [r]}W_{a_j^{(x)},y_j}
    =
    \sum_{\substack{f:\A\to\Y\\f(\ba)=\by}}Q_{x,f}=O,
\end{align*}
since $\ba\in R(\A)$ and $\by\not\in R(\Y)$. This means that the assignment $\mathfrak{W}$ is perfect for tuples $\ba^{(x)}\in R(\Lambda\X)$ of type $(i)$.

Consider now a tuple of type $(ii)$---i.e., a tuple $\bb^{(\bx)}\in R(\Lambda\X)$ for some $T\in\tau$ of arity $t$, $\bx\in T(\X)$, and $\bb\in R(\B_{T})$. For $\ell:\B_{T}\to\Y$, define $Q_{\bx,\ell}=\prod_{i\in [t]}Q_{x_i,\ell\circ\epsilon_{i,T}}$. We know from Claims~\ref{claim_1248_2311} and~\ref{eqn_1348_2311} that the family $\{Q_{\bx,\ell}\}_{\ell:\B_{T}\to\Y}$ contains mutually orthogonal projectors. Hence, we can proceed in the same way as in the previous case:
\begin{align*}
    \prod_{j\in [r]}W_{b_j^{(\bx)},y_j}
    =
    \prod_{j\in [r]}\;\sum_{\substack{\ell:\B_{T}\to\Y\\\ell(b_j)=y_j}}Q_{\bx,\ell}
    =
    \sum_{\substack{\ell:\B_{T}\to\Y\\\ell(\bb)=\by}}Q_{\bx,\ell}=O,
\end{align*}
since $\bb\in R(\B_{T})$ and $\by\not\in R(\Y)$. Thus, $\mathfrak{W}$ is perfect for tuples of type $(ii)$, too.
\end{proof}

\begin{claim} The quantum assignment $\mathfrak{W}$ is $k$-compatible.
\end{claim}
\begin{proof}
    Take four vertices $a^{(x)},\tilde a^{(\tilde x)},b^{(\bx)},\tilde b^{(\tilde\bx)}$ of $\Lambda\X$, where $x,\tilde x\in X$, $a,\tilde a\in A$, $\bx,\tilde\bx\in T(\X)$ for some symbol $T\in\tau$ or arity $t$, and $b,\tilde b\in B_{T}$. Take also two indices $i,j\in [t]$. It is easy to check from Definition~\ref{defn_left_pultr_functor} that
    \begin{align*}
        \dist_\X(x,\tilde x)&\leq \dist_{\Lambda\X}(a^{(x)},\tilde a^{(\tilde x)}),\\
        \dist_\X(x_i,\tilde x_j)&\leq \dist_{\Lambda\X}(b^{(\bx)},\tilde b^{(\tilde\bx)}),\\
        \dist_\X(x,\tilde x_j)&\leq\dist_{\Lambda\X}(a^{(x)},\tilde b^{(\tilde \bx)}),\;\;\mbox{and}\\
        \dist_{\X}(x_i,\tilde x)&\leq\dist_{\Lambda\X}(b^{(\bx)},\tilde a^{(\tilde x)}).
    \end{align*}
    Therefore, if $\dist_{\Lambda\X}(a^{(x)},\tilde a^{(\tilde x)})\leq k$ (resp. $\dist_{\Lambda\X}(b^{(\bx)},\tilde b^{(\tilde\bx)})\leq k$, $\dist_{\Lambda\X}(a^{(x)},\tilde a^{(\tilde x)})\leq k$, $\dist_{\Lambda\X}(b^{(\bx)},\tilde a^{(\tilde x)})\leq k$), the $k$-compatibility of the assignment $\mathfrak{Q}$ implies that the projectors $Q_{x,f}$, $Q_{\tilde x,\tilde f}$, $Q_{x_i,\ell\circ\epsilon_{i,T}}$, and $Q_{\tilde x_{\tilde j},\tilde \ell\circ\epsilon_{j,T}}$ appearing in the correponding definitions~\eqref{eqn_1615_2211_a} and~\eqref{eqn_1615_2211_b} commute with each other. As consequence, for each $y,y'\in Y$, we have $[W_{a^{(x)},y},W_{\tilde a^{(\tilde x)},y'}]=O$ (resp. $[W_{b^{(\bx)},y},W_{\tilde b^{( \tilde \bx)},y'}]=O$, $[W_{a^{(x)},y},W_{\tilde b^{( \tilde \bx)},y'}]=O$, $[W_{b^{(\bx)},y},W_{\tilde a^{( \tilde x)},y'}]=O$). This shows that the assignment $\mathfrak{W}$ is $k$-compatible, thus concluding the proof of the claim.
\end{proof}
In summary, we have proved that the set $\mathfrak{W}$ defined in~\eqref{eqn_1727_2211} yields a perfect $k$-compatible quantum assignment for the $\Lambda\X,\Y$ game, thus showing that $\Lambda\X\leadsto^k\Y$.
\end{proof}

In the classical setting, the adjunction between $\Gamma$ and $\Lambda$ straightforwardly implies that $\Gamma$ and $\Lambda$ are both functors between the thin categories $\mathbf{Str}(\rho)$ and $\mathbf{Str}(\tau)$, in that they satisfy the implications $\X\to\Y\;\Rightarrow\;\Gamma\X\to\Gamma\Y$ for each pair of $\rho$-structures $\X,\Y$, and $\X'\to\Y'\;\Rightarrow\;\Lambda\X'\to\Lambda\Y'$ for each pair of $\tau$-structures $\X',\Y'$. 
The adjunction relation of Theorems~\ref{thm_adjunction_pultr_1} and~\ref{thm_quantum_adjunction_faithful} allows establishing analogous results for locally compatible quantum strategies, provided that the Pultr template is connected or faithful, respectively. These functoriality properties shall be useful in the next section.

\begin{cor}
\label{cor_gamma_quantum_functor}
    Let $\mathfrak{T}$ be a connected $(\rho,\tau)$-Pultr template, let $\X$ and $\Y$ be $\rho$-structures, let $k\in\N$, let $k'=(k+1)\cdot\diam(\mathfrak{T})$, and suppose that $\X\leadsto^{k'}\Y$. Then $\Gamma\X\leadsto^{k}\Gamma\Y$.
\end{cor}
\begin{proof}
    Applying Theorem~\ref{thm_pultr_adjoints} to $\Gamma\X\to\Gamma\X$, we find $\Lambda\Gamma\X\to\X$. Composing the latter homomorphism with the quantum homomorphism $\X\leadsto^{k'}\Y$ via Proposition~\ref{prop_compose_classic_quantum_homos} yields $\Lambda\Gamma\X\leadsto^{k'}\Y$. We can then apply Theorem~\ref{thm_adjunction_pultr_1} to conclude that $\Gamma\X\leadsto^k\Gamma\Y$, as claimed.
\end{proof}

\begin{cor}
\label{cor_left_pultr_functor}
    Let $\mathfrak{T}$ be a faithful $(\rho,\tau)$-Pultr template, let $\X$ and $\Y$ be $\tau$-structures, let $k\in\N$, and suppose that $\X\leadsto^{k}\Y$. Then $\Lambda\X\leadsto^{k}\Lambda\Y$.
\end{cor}
\begin{proof}
 From $\Lambda\Y\to\Lambda\Y$, we obtain $\Y\to\Gamma\Lambda\Y$ via Theorem~\ref{thm_pultr_adjoints}. Composing the latter homomorphism with the quantum homomorphism $\X\leadsto^{k}\Y$ via Proposition~\ref{prop_compose_classic_quantum_homos} yields $\X\leadsto^{k}\Gamma\Lambda\Y$. We can then apply Theorem~\ref{thm_quantum_adjunction_faithful} to conclude that $\Lambda\X\leadsto^k\Lambda\Y$, as claimed.   
\end{proof}

\begin{rem}
\label{rem_pultr_preserves_H}
We conclude this section by noting that the proofs of Theorems~\ref{thm_adjunction_pultr_1} and~\ref{thm_quantum_adjunction_faithful} show that Pultr reductions preserve the Hilbert space $\HH$ witnessing perfect (locally compatible) quantum assignments. As we shall see later, together with a result from~\cite{simul_banakh}, this implies that the dimension of the shared quantum system witnessing perfect quantum assignments in Conjecture~\ref{conj_d_to_1_khot_quantum} must be unbounded, see Remark~\ref{rem_higher_consistency_gap}.
\end{rem}

\section{The quantum chromatic gap via $d$-to-$1$ games}
\label{sec_reductions}

The goal of this section is to prove that all links in a chain of reductions from $d$-to-$1$ games to $3$ vs. $\mathcal{O}(1)$ colouring preserve quantum completeness and, thus, can be used to transfer pseudo-telepathy. We divide the task into three phases, described in Sections~\ref{subsec_preprocessing},~\ref{subsec_GS_reduction}, and~\ref{subsec_line_digraph}, respectively. The results are then combined in Section~\ref{subsect_putting_all_together}, thus yielding a proof of the main Theorem~\ref{thm_main}. 

\subsection{From $d$-to-$1$ games to $d$-to-$d$ games} 
\label{subsec_preprocessing}
As we shall see in the next section, the known reductions from label cover to graph colouring based on a Markov chain having specific spectral properties require that the constraints in the input label-cover instance be $d$-to-$d$ rather than $d$-to-$1$. Hence, 
the first step is to reduce $d$-to-$1$ instances to $d$-to-$d$ instances. To that end, we make use of a sequence of transformations from~\cite{Dinur09:sicomp}.

Take $m,d\in\N$, and let $A$ be an alphabet set of size $md$. We say that a set $S\subseteq A^2$ is \textit{$d$-to-$d$} if the Boolean $md\times md$ matrix encoding $S$ can be written as $P_\mu(I_m\otimes J_d)P_\nu$ for two permutations $\mu,\nu$ of $[md]$, where $I$ and $J$ are the identity and all-one matrices, respectively, and $P_\mu$ and $P_\nu$ are the permutation matrices associated with $\mu$ and $\nu$, respectively.
We say that a label-cover instance $\Phi=(X,E,A,\pi,\phi)$ is \textit{$d$-to-$d$} if $|A|=md$ for some $m\in\N$, and for each $\bx\in E$ the set $\phi_\bx$ is $d$-to-$d$.

In order for the reduction in Section~\ref{subsec_GS_reduction} to go through, 
we shall need a slightly different notion of classical satisfiability 
for label-cover instances that does not make use of weights. Exactly the same notion is used in various approximation hardness results based on the UGC or $d$-to-$1$ games Conjecture, see in particular~\cite{khot2008vertex,Dinur09:sicomp}.
Given an integer $t\in\N$, 
we denote by ${A\choose \leq t}$ the set of subsets of the alphabet $A$ of size at most $t$.
A \textit{$t$-assignment} for a label-cover instance $\Phi$ is a function $f:X\to{A\choose \leq t}$. We say that $f$ satisfies a constraint $\bx=(x_1,x_2)\in E$ if there exist $a\in f(x_1)$, $b\in f(x_2)$ such that $\phi_\bx(a,b)=1$. 
We define
\begin{align*}
    \isat_t(\Phi)
    =
    \frac{1}{|X|}\cdot \max_{S\subseteq X}\Big\{|S|\suchThat\exists\mbox{ $t$-assignment for $\Phi$ satisfying all constraints induced by $S$}\Big\}.
\end{align*}

Recall that, following~\cite{Khot02stoc}, we have defined label cover in Section~\ref{subsec_CSP_predicates} as a \textit{weighted} CSP, where edges carry a probability distribution $\pi$ that appears in the definition of $\sat(\Phi)$. 
Since the probability distribution $\pi$ does not appear in the definition of $\isat_t(\Phi)$ (and, as discussed in Section~\ref{sec_classical_quantum_CSPs}, it also does not appear in the definition of perfect quantum assignments), we shall
now consider \textit{unweighted} label-cover instances, where $\pi$ is assumed to be, say, the uniform distribution over $E$.
We denote an unweighted label-cover instance simply by $\Phi=(X,E,A,\phi)$. 

The next result shows that, assuming Conjecture~\ref{conj_d_to_1_khot_quantum}, there exist $d$-to-$d$ instances admitting a perfect locally compatible quantum assignment but arbitrarily small $\isat_t$ value. Its proof is deferred to Appendix~\ref{app_dinur_reductions}.

\begin{prop}
\label{prop_d_to_d_after_dinur}
    Suppose Conjecture~\ref{conj_d_to_1_khot_quantum} holds for some $d\geq 2$. 
Then for each $\epsilon>0$ and each $k,t\in\N$ there exists a $d$-to-$d$ instance $\Phi$ (on some alphabet size $n$) such that $\Phi$ admits a perfect $k$-compatible quantum assignment but
$\isat_t(\Phi)<\epsilon$.
\end{prop}

\subsection{From $d$-to-$d$ games to $2d$ vs. $\mathcal{O}(1)$ colouring}
\label{subsec_GS_reduction}

We now present a reduction described in~\cite{Dinur09:sicomp} (in the case $d=2$) and in~\cite{GS20:icalp} (in the general case $d\geq 2$).
The starting point is the following, spectral-graph-theoretic result.

\begin{thm}[\cite{Dinur09:sicomp,GS20:icalp}]
\label{thm_transistion_matrix}
For each $2\leq d\in\N$ there exists a Markov chain on state set $\Omega=[2d]^d$ whose transition matrix $T=(t_{ij})$ has the following properties:
    \begin{itemize}
        \item $T$ is symmetric, i.e., the transition probabilities $t_{\bx,\by}$ and $t_{\by,\bx}$ are equal for each $\bx,\by\in\Omega$;
        \item the eigenvalue $1$ has multiplicity one in the spectrum of $T$, and all other eigenvalues have modulus strictly smaller than $1$;
        \item for $\bx=(x_1,\dots,x_d),\by=(y_1,\dots,y_d)\in\Omega$, $t_{\bx,\by}=0$ when $\{x_1,\dots,x_d\}\cap\{y_1,\dots,y_d\}\neq\emptyset$.
    \end{itemize}
\end{thm}

Let $\Phi=(X,E,[n],\phi)$ be an (unweighted) $d$-to-$d$ instance with alphabet $[n]$, where $n=md$ for some $m\in\N$.
Recall that, for each $\bx\in E$, the corresponding set of admitted labellings $\phi_\bx$ is $d$-to-$d$, meaning that there exists two permutations $\mu,\nu\in\Sym_{n}$ such that $\phi_\bx$ is encoded by the matrix $P_\mu(I_m\otimes J_d)P_\nu$. We shall call $\mu,\nu$ the permutations associated with $\bx$. 
On input $\Phi$, the reduction returns a digraph $\eta\Phi$ defined as follows. The vertex set is $\{(x,z):x\in X,\,z\in[2d]^n\}$.
The edge set contains those pairs $((x,z),(x',z'))$ of vertices for which
\begin{itemize}
    \item[$(i)$] 
    $(x,x')$ forms an edge in $E$;
    \item [$(ii)$] for each $i\in [m]$, the $(\alpha(i),\beta(i))$-th entry of the transition matrix $T$ of Theorem~\ref{thm_transistion_matrix} is nonzero, where $\mu,\nu$ are the permutations associated with the edge $(x,x')$ and
    \begin{equation}
        \label{eqn_1248_2411}
         \begin{aligned}
        \alpha(i)&=(z\circ\mu(di-d+1),z\circ\mu(di-d+2),\dots,z\circ\mu(di))\in [2d]^d,\\
        \beta(i)&=(z'\circ\nu(di-d+1),z'\circ\nu(di-d+2),\dots,z'\circ\nu(di))\in [2d]^d.
    \end{aligned}
    \end{equation}
\end{itemize}

The following soundness result was proved in~\cite{Dinur09:sicomp} for the case $d=2$, and in~\cite{GS20:icalp} for the case $d\geq 3$.

\begin{thm}[\cite{Dinur09:sicomp,GS20:icalp}]
\label{thm_soundess_third_reduction}
For each $\epsilon>0$ there exist $\epsilon'>0$, $t\in\N$ such that, given a $d$-to-$d$ instance $\Phi$, there are no independent sets of relative size $\epsilon$ in $\eta\Phi$ whenever $\isat_t(\Phi)<\epsilon'$. 
\end{thm}

We shall now prove that $\eta$ is quantum-complete.

\begin{thm}
\label{thm_quantum_completeness_dinur}
    Let $\Phi$ be a $d$-to-$d$ instance admitting a perfect $k$-compatible quantum assignment for some $k\in\N$. Then $\eta\Phi\leadsto^k\K_{2d}$.
\end{thm}

To establish Theorem~\ref{thm_quantum_completeness_dinur}, we use the results of Section~\ref{sec_pultr_functors}.
First of all, it shall be useful to see $\eta$ as the left Pultr functor associated with a suitable
Pultr template.
Let $\rho$ be the signature containing a unique, binary symbol $R$, and let $\tau$ be the signature containing a binary symbol $T_{\mu,\nu}$ for each choice of $\mu,\nu\in\Sym_n$.
Consider the $(\rho,\tau)$-Pultr template $\mathfrak{T}$ described as follows:
\begin{itemize}
    \item $\A$ is the union of $[2d]^n$ isolated vertices.
    \item For each symbol $T_{\mu,\nu}\in\tau$, $\B_{T_{\mu,\nu}}$ is a bipartite digraph with domain $B_1\sqcup B_2$, where $|B_1|=|B_2|=[2d]^n$. The edges of $\B_{T_{\mu,\nu}}$ (i.e., the elements of $R(\B_{T_{\mu,\nu}})$) are defined as follows. Given $z\in B_1$ and $z'\in B_2$, the pair $(z,z')$ forms an edge if and only if for each $i\in [m]$, the $(\alpha(i),\beta(i))$-th entry of the transition matrix $T$ is nonzero, with $\alpha(i)$ and $\beta(i)$ as in~\eqref{eqn_1248_2411}.
    \item For each symbol $T_{\mu,\nu}\in\tau$ and each $i\in [2]$, the homomorphism $\epsilon_{i,T_{\mu,\nu}}:\A\to\B_{T_{\mu,\nu}}$ sends $\A$ to the first or the second part of the domain of $\B_{T_{\mu,\nu}}$, depending on whether $i=1$ or $i=2$.
\end{itemize}

Let $\Lambda_\mathfrak{T}$ be the left Pultr functor corresponding to $\mathfrak{T}$. 
Given a $d$-to-$d$ instance $\Phi=(X,E,[n],\phi)$, let $\X_\Phi$ 
and $\dGames$ be the corresponding relational structures, as described in Section~\ref{subsec_CSP_homomorphisms}.
To spell this out,  $\X_\Phi$ 
is the $\tau$-structure having domain $X$ such that, for each $\mu,\nu\in\Sym_n$, 
\begin{align*}
T_{\mu,\nu}(\X_\Phi)=\{\bx\in E\suchThat \mu,\nu\mbox{ are the permutations associated with }\phi_\bx\}. 
\end{align*}
Moreover, $\dGames$ is the $\tau$-structure having domain $[n]$ and whose relations are defined as follows. Fix the matrix $H=[h_{ij}]=I_{m}\otimes J_d$.
For $\mu,\nu\in\Sym_n$, we let ${T_{\mu,\nu}}({\dGames})=\{(a,b)\in [n]^2:h_{\mu^{-1}(a),\nu^{-1}(b)}=1\}$. 
It is then straightforward to check that 
\begin{align}
\label{eqn_1349_1812}
\eta\Phi=\Lambda_{\mathfrak{T}}\X_\Phi.
\end{align}
The proof of the next fact uses the same ideas as in~\cite{Dinur09:sicomp,GS20:icalp}, though stated in the language of Pultr functors as needed in our setting.
 
%
%

\begin{prop}
\label{prop_simple_lambda_dinur}
    $\Lambda_{\mathfrak{T}}\dGames\to\K_{2d}$.
\end{prop}
\begin{proof}

Take a vertex $a^{(x)}$ of $\Lambda_{\mathfrak{T}}\dGames$ where
    $a\in [2d]^n$ and $x\in [n]$, and define 
        $\xi(a^{(x)})=a(x)$.
    Also, for $\mu,\nu\in\Sym_n$, take a vertex $z^{(\bx)}$ of $\Lambda_{\mathfrak{T}}\dGames$, where $z\in\B_{T_{\mu,\nu}}$ and $\bx=(x_1,x_2)\in T_{\mu,\nu}(\dGames)$. Recall that the domain of $\B_{T_{\mu,\nu}}$ is split into two disjoint parts $B_1$ and $B_2$, each of size $[2d]^n$.    Let $\xi(z^{(\bx)})=z(x_1)$ if $z\in B_1$, and $\xi(z^{(\bx)})=z(x_2)$ if $z\in B_2$.
    We claim that the map $\xi$ yields the required homomorphism from $\Lambda_{\mathfrak{T}}\dGames$ to $\K_{2d}$.

    First of all, we need to make sure that $\xi$ respects the equivalence relation of Definition~\ref{defn_left_pultr_functor} and is thus well defined. Take $\mu,\nu\in\Sym_n$, $\bx\in T_{\mu,\nu}(\dGames)$, $j\in [2]$, and $a\in A=[2d]^n$. The two vertices $a^{(x_j)}$ and $[\epsilon_{j}(a)]^{(\bx)}$ are identified in $\Lambda_{\mathfrak{T}}\dGames$, so we need to prove that
    \begin{align}
    \label{eqn_1119_2411}
        \xi(a^{(x_j)})
        =
        \xi([\epsilon_{j}(a)]^{(\bx)}).
    \end{align}
    (We are dropping the index $T_{\mu,\nu}$ from $\epsilon_{j,T_{\mu,\nu}}$ to improve readability.)
    It follows from the definition of $\xi$ and $\epsilon_j$ that $\xi([\epsilon_{j}(a)]^{(\bx)})=a(x_1)$ if $j=1$, and $\xi([\epsilon_{j}(a)]^{(\bx)})=a(x_2)$ if $j=2$. Hence, $\xi([\epsilon_{j}(a)]^{(\bx)})=a(x_j)$. Since $\xi(a^{(x_j)})=a(x_j)$,~\eqref{eqn_1119_2411} follows.

    We now show that $\xi$ preserves the edge relation and is thus a homomorphism from $\Lambda_{\mathfrak{T}}\dGames$ to $\K_{2d}$. $\A$ has no edges, so we only need to consider the edges of $\Lambda_{\mathfrak{T}}\dGames$ coming from $\B_{T_{\mu,\nu}}$. Take $\bx\in T_{\mu,\nu}(\dGames)$ and $(z,z')\in R(\B_{T_{\mu,\nu}})$, so $(z^{(\bx)},z'^{(\bx)})\in R(\Lambda_{\mathfrak{T}}\dGames)$. (Recall that $R$ denotes the unique, binary symbol of the signature $\rho$.)
    We have $\xi(z^{(\bx)},z'^{(\bx)})=(z(x_1),z'(x_2))$. Since $(z,z')$ forms an edge of $\B_{T_{\mu,\nu}}$, the $(\alpha,\beta)$-th entry of the transition matrix $T$ used to describe the Pultr template $\mathfrak{T}$ is nonzero for each $i\in [m]$, where $\alpha$ and $\beta$ are as in~\eqref{eqn_1248_2411}. It then follows from the third property of $T$ that $z(x_1)\neq z'(x_2)$, which means that $(z(x_1),z'(x_2))\in R(\K_{2d})$, as needed.
\end{proof}

\begin{proof}[Proof of Theorem~\ref{thm_quantum_completeness_dinur}]
Let $\Phi$ be an instance of $d$-to-$d$ games, and let $\X_\Phi$ be the corresponding $\tau$-structure as defined above, so that the equation~\eqref{eqn_1349_1812} holds.
Recall from Section~\ref{subsec_CSP_homomorphisms} that 
$\Phi$ admitting a perfect $k$-compatible quantum assignment is equivalent to $\X_\Phi\leadsto^k\dGames$.
Observe that the Pultr template $\mathfrak{T}$ is faithful. 
Using Corollary~\ref{cor_left_pultr_functor}, we obtain $\Lambda_{\mathfrak{T}}\X_\Phi\leadsto^k\Lambda_{\mathfrak{T}}\dGames$. Composing this with the homomorphism from Proposition~\ref{prop_simple_lambda_dinur} via Proposition~\ref{prop_compose_classic_quantum_homos}, we conclude that $\eta\Phi=\Lambda_{\mathfrak{T}}\X_\Phi\leadsto^k\K_{2d}$, as required.
\end{proof}

\subsection{The line-digraph reduction}
\label{subsec_line_digraph}

The results of Sections~\ref{subsec_preprocessing} and~\ref{subsec_GS_reduction} yield a $2d$ vs. $\mathcal{O}(1)$ quantum chromatic gap, conditional to the pseudo-telepathy of $d$-to-$1$ games. The next task is to reduce the completeness parameter from $2d$ to $3$. To that end, we make use of a standard graph-theoretical reduction, which we now define (recall also Example~\ref{ex_some_pultr_functors}). 

\begin{defn}
Let $\X$ be a digraph having vertex set $X$ and edge set $R(\X)$.
The \emph{line digraph} of $\X$ is the digraph $\delta\X$ whose vertex set is $R(\X)$ and whose edge set is $$\{((x,y),(y,z))\suchThat (x,y),(y,z)\in R(\X)\}.$$
\end{defn}

The line-digraph reduction is known to change the chromatic number of digraphs in a controlled way, as described next.
For $n\in\N$, let $\alpha(n)=2^n$ and $\beta(n)={n\choose \floor{n/2}}$. Observe that $\alpha(n)\geq \beta(n)$ for each $n$.   
\begin{thm}[\cite{HarnerE72,poljak1981arc}]
\label{prop_chromatic_number_line_digraph}
Let $\X$ be a digraph and let $n\in\N$.
If $\delta\X\to\K_n$, then $\X\to\K_{\alpha(n)}$;
if $\X\to\K_{\beta(n)}$, then $\delta\X\to\K_n$.
\end{thm}

\begin{thm}
\label{thm_quantum_completeness_arc_digraph}
    Let $\X$ and $\Y$ be digraphs and let $k\in\N$. If $\X\leadsto^{2k+2}\Y$, then $\delta\X\leadsto^{k}\delta\Y$.
\end{thm}
\begin{proof}
Recall from Example~\ref{ex_line_digraph_pultr} that $\rho$ is the central Pultr functor corresponding to a connected Pultr template $\mathfrak{T}$. In particular, notice that $\diam(\mathfrak{T})=2$. Hence, the result immediately follows from Corollary~\ref{cor_gamma_quantum_functor}.
\end{proof}

\subsection{Putting all together}
\label{subsect_putting_all_together}
Combining the results of Subsections~\ref{subsec_preprocessing},~\ref{subsec_GS_reduction}, and~\ref{subsec_line_digraph}, we can now conclude the proof of the first main result of this paper.

\begin{thm*}
[Theorem~\ref{thm_main} restated]
Suppose that Conjecture~\ref{conj_d_to_1_khot_quantum} holds for some $d\geq 2$. Then, for any $c\in\N$, there exists a graph with quantum chromatic number $3$ and classical chromatic number larger than $c$.
\end{thm*}
\begin{proof}
Take an integer $d\geq 2$ such that Conjecture~\ref{conj_d_to_1_khot_quantum} holds for $d$.
Recall the functions $\alpha$ and $\beta$ defined in Section~\ref{subsec_line_digraph}.
For $i\in \N$, we denote by $\alpha^{(i)}$ (resp. $\beta^{(i)}$) the function obtained by iterating $\alpha$ (resp. $\beta$) $i$-many times.
Observe that $\beta(n)>n$ whenever $n\geq 4$. Hence, there exists some $j\in\N$ such that $\beta^{(j)}(4)\geq 2d$. Let $i=j+2$, and define $k'=3\cdot 2^i-2$. 
Assume, without loss of generality, that $c\geq 2$.
Define $\epsilon=\frac{1}{\alpha^{(i)}(c)}$, and choose $\epsilon'>0$, $t\in\N$ witnessing the truth of Theorem~\ref{thm_soundess_third_reduction} applied to $\epsilon$. Using Proposition~\ref{prop_d_to_d_after_dinur}, find a $d$-to-$d$ instance $\Phi$ such that 
$\Phi$ admits a perfect $k'$-compatible quantum assignment
but $\isat_t(\Phi)<\epsilon'$.
Consider the digraph $\delta^{(i)}\eta\Phi$, where $\delta^{(i)}$ denotes the $i$-fold application of the line-digraph operator $\delta$ of Section~\ref{subsec_line_digraph} and $\eta$ is the reduction from Section~\ref{subsec_GS_reduction}. Let also $\GG$ be the undirected graph obtained by replacing each directed edge of $\delta^{(i)}\eta\Phi$ by an undirected edge. We claim that $\GG$ has quantum chromatic number $3$ and classical chromatic number larger than $c$, which would conclude the proof.

\emph{(Classical) soundness}.\quad Theorem~\ref{thm_soundess_third_reduction} implies that $\eta\Phi$ has no independent sets of relative size $\epsilon$. This means that the classical chromatic number of $\eta\Phi$ is larger than $\frac{1}{\epsilon}=\alpha^{(i)}(c)$; hence, $\eta\Phi\not\to\K_{\alpha^{(i)}(c)}$. By repeatedly applying (the contrapositive of the first part of) Theorem~\ref{prop_chromatic_number_line_digraph}, we deduce that $\delta^{(i)}\eta\Phi\not\to\K_c$. Since the classical chromatic number of a digraph is equal to that of the corresponding undirected graph, we obtain $\GG\not\to\K_c$.

\emph{(Quantum) completeness}.\quad Theorem~\ref{thm_quantum_completeness_dinur} guarantees that $\eta\Phi\leadsto^{k'}\K_{2d}$. Since $\K_{2d}\to\K_{\beta^{(j)}(4)}$, Proposition~\ref{prop_compose_classic_quantum_homos} ensures that $\eta\Phi\leadsto^{k'}\K_{\beta^{(j)}(4)}$. Consider now the function $f:\N\to\N$ given by $n\mapsto 2n+2$, and observe that its $i$-fold iteration applied to the input $1$ is $f^{(i)}(1)=3\cdot 2^i-2=k'$. Hence, a repeated application of Theorem~\ref{thm_quantum_completeness_arc_digraph} yields
\begin{align}
    \label{eqn_1933_1912}
    \delta^{(i)}\eta\Phi\leadsto^1\delta^{(i)}\K_{\beta^{(j)}(4)}.
\end{align}
Starting from the identity homomorphism $\K_{\beta^{(j)}(4)}\to\K_{\beta^{(j)}(4)}$, a repeated application of (the second part of) Theorem~\ref{prop_chromatic_number_line_digraph} yields $\delta^{(j)}\K_{\beta^{(j)}(4)}\to\K_4$. Then, the functoriality of $\delta$ ensures that
\begin{align}
    \label{eqn_1933_1912_b}\delta^{(i)}\K_{\beta^{(j)}(4)}
    =
    \delta^{(2)}\delta^{(j)}(\K_{\beta^{(j)}(4)})\to\delta^{(2)}\K_4\to\K_3,
\end{align}
where the last homomorphism comes from the fact that $\delta^{(2)}\K_4$ is classically $3$-colourable---as was shown in~\cite{zhu1998survey} (see also~\cite{poljak1991coloring}). Combinining~\eqref{eqn_1933_1912} and~\eqref{eqn_1933_1912_b} via Proposition~\ref{prop_compose_classic_quantum_homos}, we obtain $\delta^{(i)}\eta\Phi\leadsto^1\K_3$, which implies that $\delta^{(i)}\eta\Phi\leadsto^0\K_3$.
Like in the classical case, making the edges of a digraph undirected does not affect the quantum chromatic number. Hence, we deduce that $\GG\leadsto^0\K_3$. Finally, as noted in~\cite{MancinskaR16}, a graph has quantum chromatic number $2$ precisely when it is bipartite. Since $\GG\not\to\K_c$ and $c\geq 2$, $\GG$ is not bipartite.
It follows that the quantum chromatic number of $\GG$ is exactly $3$, which concludes the proof.
\end{proof}

\begin{rem}
\label{rem_higher_consistency_gap}
    We point out that the proof of Theorem~\ref{thm_main} can be easily modified to enforce local compatibility of the perfect quantum assignment witnessing that the graph is quantum $3$-colourable. More precisely, conditional to Conjecture~\ref{conj_d_to_1_khot_quantum}, for each $k\in\N$ there exist graphs $\GG$ having arbitrarily large classical chromatic number and such that $\GG\leadsto^k\K_3$. To prove this, one simply needs to set $k'=3\cdot 2^{i+k}-2$ instead of $k'=3\cdot 2^{i}-2$ in the proof of Theorem~\ref{thm_main}. This choice guarantees that~\eqref{eqn_1933_1912} can be replaced by the stronger
    \begin{align*} \delta^{(i)}\eta\Phi\leadsto^k\delta^{(i)}\K_{\beta^{(j)}(4)},
\end{align*}
which implies that $\GG\leadsto^k\K_3$.

This fact, together with an upper bound on the quantum chromatic gap in terms of the number of shared qubits in the quantum strategy recently proved in~\cite{simul_banakh}, has the interesting consequence that the Hilbert spaces necessary for the validity of Conjecture~\ref{conj_d_to_1_khot_quantum} cannot have bounded dimension.
More precisely, it was shown in~\cite[Theorem~29]{simul_banakh} that, for each graph $\GG$ and each $n\in\N$, if $\GG\leadsto^1\K_n$ on a Hilbert space of dimension $p$, then $\GG\to\K_{n\cdot \alpha_p}$ with $\alpha_p=(2+\mathcal{o}(1))^{p-1}$. 
Now, all reductions from $d$-to-$1$ games to $3$ vs. $\mathcal{O}(1)$ colouring described in Section~\ref{sec_reductions}  preserve the Hilbert space $\HH$ witnessing perfect quantum assignments, as discussed in Remarks~\ref{rem_pultr_preserves_H} and~\ref{rem_appendix_reductions_preserve_H}. It follows that, in Conjecture~\ref{conj_d_to_1_khot_quantum}, the dimension of the Hilbert spaces witnessing the presence of a perfect quantum assignment for the classically $\epsilon$-unsatisfiable instance $\Phi$ must approach infinity as  $\epsilon$ approaches $0$.
\end{rem}

\section{The quantum  chromatic gap via $3$XOR games}
\label{sec_quantum_chromatic_gap_3xor}
The Mermin--Peres ``magic square'' construction gives a well-known example of a pseudo-telepathic system of Boolean equations. The classical value of the system is strictly smaller than $1$; however, by sharing a pair of entangled Bell states and measuring it via PVMs arising from combinations of Pauli matrices, the corresponding non-local game can be won with probability $1$ by two cooperating non-communicating players~\cite{mermin1990simple,mermin1993hidden,peres1990incompatible}.
The Mermin--Peres system is an instance of $3$XOR, i.e., the CSP expressing systems of Boolean equations involving three variables each. Furthermore, the  Mermin--Peres system is \textit{$2$-regular}, in that every variable appears in at most (in fact, precisely) $2$ equations, and two distinct equations share at most one variable.
Systems of linear equations can be naturally interpreted as instances of ``many-to-$1$'' games: For every partial assignment to some equation, there exists precisely one consistent assignment to each of the variables appearing in the equation.
Hence, a naive approach to proving Conjecture~\ref{conj_d_to_1_khot_quantum}---and, thus, establishing Theorem~\ref{thm_main} unconditionally---could be to consider a parallel-repetition version of the Mermin--Peres magic square game. In that way, the classical value of the game is decreased while preserving the existence of a perfect quantum assignment (which is obtained by simply taking the tensor product of PVMs corresponding to one iteration of the game). 

Unfortunately, it is not hard to see that this approach cannot work: Increasing the number of repetitions of the game generates label-cover instances that are $d$-to-$1$ for increasingly large $d$. On the other end, in order to prove Theorem~\ref{thm_main} by transferring pseudo-telepathy from the instances in Conjecture~\ref{conj_d_to_1_khot_quantum}, one would need to reduce the soundness parameter $\epsilon$ to zero while leaving $d$ \textit{fixed}.
Therefore, we consider a stronger version of repetition for $3$XOR games. It shall be convenient to present it in the form of a non-local game. Let $S$ be a system of linear equations over the Boolean field $\F_2$, and fix an integer $n\in\N$. Alice receives from the verifier a list $\bu=(u_1,\dots,u_n)$ of equations from $S$, and replies with an assignment $\vartheta:\widetilde\bu\to\F_2$, where $\tilde\bu$ is the set of all variables appearing in the equations from $\bu$. Similarly, Bob receives from the verifier a list $\bw=(w_1,\dots,w_n)$ of equations from $S$ and replies with an assignment $\vartheta':\widetilde\bw\to\F_2$. Alice and Bob win the game if $\vartheta$ respects the equations in $\bu$, $\vartheta'$ respects the equations in $\bw$, and the replies are consistent, in the sense that 
\begin{align}
\label{eqn_0801_1029}
\vartheta|_{\widetilde\bu\cap\widetilde{\bw}}=\vartheta'|_{\widetilde\bu\cap\widetilde{\bw}}. 
\end{align}
We denote this game by $\mathcal{G}(S,n)$.
Note that the last requirement would not be enforced in a standard parallel repetition of $3$XOR, where Alice's and Bob's replies need not be consistent across different iterations of the game. 
%


Let $s(\bu)$ denote the set of assignments $\vartheta:\widetilde\bu\to\F_2$ such that $\vartheta$ satisfies all equations in $\bu$.
Using Proposition~\ref{thm_quantum_assignments_as_projectors}, we have that a perfect quantum strategy $\mathfrak{Q}$ for the $\mathcal{G}(S,n)$ game is given by projectors $Q_{\bu,\vartheta}$ onto a Hilbert space $\HH$ for each $n$-tuple $\bu$ of equations and each $\vartheta\in s(\bu)$, such that
\begin{itemize}
    \item $\sum_{\vartheta\in s(\bu)}Q_{\bu,\vartheta}=\id_\HH$ for each $\bu$, and
    \item $Q_{\bu,\vartheta}\cdot Q_{\bw,\vartheta'}=O$ whenever $\vartheta|_{\widetilde\bu\cap\widetilde{\bw}}\neq\vartheta'|_{\widetilde\bu\cap\widetilde{\bw}}$.
\end{itemize}
We shall further require that the strategy is $1$-compatible; i.e., 
\begin{itemize}
    \item $[Q_{\bu,\vartheta},Q_{\bw,\vartheta'}]=O$ whenever $\tilde\bu\cap\tilde\bw\neq \emptyset$. 
\end{itemize}

As usual, we let $\sat(S)$ denote the maximum fraction of equations simultaneously satisfiable by a classical assignment. 
Moreover, as a generalisation of the regularity property of the Mermin--Peres construction, a system of linear equations is said to be \textit{regular} for some $p\in\N$ (which we consider as a fixed constant) if $(i)$ every variable appears in at most $p$ equations, and $(ii)$ two distinct equations share at most one variable. The main result of this section is the following.

\begin{thm*}
[Theorem~\ref{thm_equations_to_colouring} restated]
    Suppose that there exists some $0<s^*<1$ such that, for each $n\in\N$, there exists a regular instance $S$ of $3$XOR for which $\sat(S)<s^*$ but $\mathcal{G}(S,n)$ admits a perfect $1$-compatible quantum assignment. Then, for each $c\in\N$, there exists a graph with quantum chromatic number at most $4$ and classical chromatic number larger than $c$.
\end{thm*}

\subsection{Description of the reduction}
\label{subsec_description_rho_one}
We now outline the reduction used in~\cite{Khot17:stoc-independent,Khot18:focs-pseudorandom,Dinur18:stoc-non-optimality,Dinur18:stoc-towards} to establish the $2$-to-$2$ Games Theorem with imperfect completeness. We shall mostly follow the presentation in~\cite[ch.~3]{minzer2022monotonicity}.

Let $S$ be a regular instance of $3$XOR. 
The final reduction $\rho$
is the composition of two reductions $\rho_1$ and $\rho_2$:
The first takes $S$ as input and returns a label-cover instance $\rho_1 S $ all of whose constraints are either $1$-to-$1$ or $2$-to-$2$; the second takes  $\rho_1 S$ as input and returns a $2$-to-$2$ instance $\rho S=\rho_2 \rho_1 S$.

Fix positive integers $n\in\N$ and $2\leq\ell\in\N$, which we consider fixed parameters of the reduction.
For any equation $u$ in $S$, we let its \textit{right-hand side} be the number $b\in\F_2$ such that $u$ is the equation $x+y+z=b$. 
Consider the set $\mathcal{U}$ containing all $n$-tuples $\bu=(u_1,\dots,u_n)$ of equations in $S$ that are \textit{legitimate}, i.e., $(i)$ all equations $u_1,\dots,u_n$ are distinct and do not share any variable, and $(ii)$ if two distinct variables $x_1,x_2$ appear in two distinct equations of $\bu$, there is no equation in $S$ in which both $x_1$ and $x_2$ appear.
Let $X$ be the set of variables appearing in the system $S$.
We shall consider the set $\F_2^X$, which we view as a $|X|$-dimensional vector space over the Boolean field $\F_2$. Given a set $X'\subseteq X$, the vector space $\F_2^{X'}$ is canonically identified with a subspace of $\F_2^X$ (where the identification simply consists in padding each vector $v\in\F_2^{X'}$ with $(|X|-|X'|)$-many zero entries).
For $x\in X$, we let $v_x$ be the one-hot encoding of $x$; i.e, the vector of $\F_2^X$ whose $x$-th entry is $1$ and all other entries are $0$.
%
Also, for each $i\in [n]$, we consider the vector $v_{u_i}=v_x+v_y+v_z\in\F_2^{X}$, where $x$, $y$, and $z$ are the variables in $u_i$.
Denote the space $\Span(v_{u_1},\dots,v_{u_n})$ by $H_\bu$. For a space $X'\subseteq X$, we say that a linear function $\psi:\F_2^{X'}\to\F_2$ \emph{respects} $\bu$ if $\psi(v_{u_i})=b_i$ for each $i\in [n]$ for which $v_{u_i}\in\F_2^{X'}$, where $b_i$ is the right-hand side of $u_i$. Fix a tuple $\bu\in\mathcal{U}$, and let $\tilde \bu$ denote the set of variables appearing in the equations of $\bu$, as usual. 
Consider the set $\mathcal{L}_\bu$ containing all vector subspaces $L\subseteq\F_2^{\tilde \bu}$ such that $\dim(L)=\ell$ and $L\cap H_\bu=\{\bzero\}$.
The next property of $\mathcal{L}_\bu$ shall be crucial for establishing quantum completeness of $\rho_1$.
\begin{lem}[\cite{Khot17:stoc-independent}]
\label{lem_extension_side_conditions}
    Choose $\bu,\bw\in\mathcal{U}$ and $L\in\mathcal{L}_\bu$. Then any linear function on $L\oplus H_\bu$ that respects $\bu$ admits a unique linear extension to $L\oplus H_\bu\oplus H_{\bw}$ that respects $\bw$.
\end{lem}

We let $\rho_1 S$ be the unweighted label-cover instance $(V,E,A,\phi)$ defined as follows.
The variable set $V$ is the set
$\{
    L\oplus H_\bu:\bu\in\mathcal{U},L\in\mathcal{L}_\bu
    \}$.
The alphabet $A$ is the set $[2^\ell]$.
Since each of the $2^\ell$-many linear functions $L\to\F_2$ can be uniquely extended to a linear function on $L\oplus H_\bu$ that respects $\bu$, the elements of $[2^\ell]$ are naturally identified with linear functions $L\oplus H_\bu\to\F_2$ respecting $\bu$.
Given a pair $e=(L\oplus H_\bu,L'\oplus H_{\bw})$ of variables, we let $e\in E$ if and only if 
$
    \dim(L\oplus H_\bu\oplus H_{\bw})
    =
    \dim(L'\oplus H_\bu\oplus H_{\bw})
    =
    \dim(L\oplus L'\oplus H_\bu\oplus H_{\bw})
$
(in which case the constraint is $1$-to-$1$), or 
$
    \dim(L\oplus H_\bu\oplus H_{\bw})
    =
    \dim(L'\oplus H_\bu\oplus H_{\bw})
    =
    \dim(L\oplus L'\oplus H_\bu\oplus H_{\bw})-1
$
(in which case, the constraint is $2$-to-$2$). In both cases, the set $\phi_e$ of admitted labellings is defined as follows. Given a linear function $\psi:L\oplus H_\bu\to\F_2$ (i.e., an element of the alphabet),
let $\psi^{(\bw)}$ be its unique linear extension to $L\oplus H_\bu\oplus H_{\bw}$ respecting $\bw$, as per Lemma~\ref{lem_extension_side_conditions}; similarly, given $\psi':L'\oplus H_{\bw}\to\F_2$, let $\psi'^{(\bu)}$ be its unique linear extension to $L'\oplus H_\bu\oplus H_{\bw}$ respecting $\bu$. The pair $(\psi,\psi')$ belongs to $\phi_e$ if and only if $\psi^{(\bw)}$ and $\psi'^{(\bu)}$ agree on the shared space $(L\oplus H_\bu\oplus H_{\bw})\cap (L'\oplus H_\bu\oplus H_{\bw})$. It is not hard to show that the resulting constraints are $1$-to-$1$ or $2$-to-$2$ according to which of the two conditions above holds.

It was shown in~\cite{Khot17:stoc-independent} that the label-cover instance $\rho_1 S$ described above satisfies a specific \textit{transitivity property}. This is used in order to perform the second step $\rho_2$ of the reduction, which discards all $1$-to-$1$ constraints of $\rho_1 S$ and assigns a suitable weight $\pi$ to the $2$-to-$2$ constraints, so that the composition $\rho S=\rho_2 \rho_1 S$ is a (weighted) $2$-to-$2$ instance. 
The details of $\rho_2$ are not relevant to the current analysis, so we do not include them here. It is sufficient for our purposes to note that $(i)$ the domains of $\rho_1 S$ and $\rho_2\rho_1 S$ coincide, and $(ii)$ while some of the constraints are discarded in the second step of the reduction, no new constraint is introduced. Both these facts follow from the description in~\cite[ch.~3]{minzer2022monotonicity}.

\subsection{Quantum completeness of $\rho$}
The reduction $\rho$ is easily seen to preserve perfect \textit{classical} completeness, in the sense that $\rho S$ admits a perfect classical assignment whenever $S$ does. In fact, $\rho$ also preserves imperfect classical completeness, in that almost satisfiable $3$XOR instances are mapped to almost satisfiable $2$-to-$2$ instances. This is crucial for the purpose of~\cite{Khot17:stoc-independent,Khot18:focs-pseudorandom,Dinur18:stoc-non-optimality,Dinur18:stoc-towards}, as H\aa stad's NP-hardness of approximating $3$XOR games~\cite{haastad2001some} only holds in the imperfect completeness regime (otherwise, the problem trivially becomes tractable). Our goal is now to show that $\rho$ preserves \textit{quantum} completeness, and can thus be used to transfer pseudo-telepathy from $3$XOR to $2$-to-$2$ games. Since $3$XOR games are pseudo-telepathic despite being tractable in the exact version, we prove perfect quantum completeness.

\begin{thm}
\label{thm_quantum_completeness_rho_minzer}
Let $S$ be a regular $3$XOR instance. If $\mathcal{G}(S,n)$ admits a perfect $1$-compatible quantum assignment, the same holds for the $2$-to-$2$ instance $\rho S$.
\end{thm}
\begin{proof}
    Let $\mathfrak{Q}=\{Q_{\bu,\vartheta}\}$ be a perfect $1$-compatible quantum assignment for $S$ over a Hilbert space $\HH$. Let $X$ be the set of variables appearing in $S$. 
    Given a function $\xi:X'\to\F_2$ for some $X'\subseteq X$, we denote by $\widehat{\xi}:\F_2^{X'}\to\F_2$ the corresponding linear function defined by $\widehat{\xi}(v_x)=\xi(x)$ for each $x\in X'$ and extended by linearity. (Note that the set $\{v_x\}_{x\in X'}$ forms a basis for ${\F_2}^{X'}$ since $v_x$ is the one-hot encoding of $x$.)
    Recall that $\mathcal{U}$ denotes the set of all legitimate $n$-tuples of equations of $S$.
    
    We first prove quantum completeness for the first part $\rho_1$ of the reduction $\rho$.
    Take a variable $L\oplus H_\bu$ of $\rho_1 S$;
    i.e., $\bu\in\mathcal{U}$ and $L\in\mathcal{L}_\bu$. Take also a label $\psi\in A$, which corresponds to a linear map $\psi:L\oplus H_\bu\to\F_2$ that respects $\bu$. Consider the linear application
    \begin{align}
    \label{eqn_0801_1740}
        W_{L\oplus H_\bu,\psi}
        =
        \sum_{\substack{\vartheta\in s(\bu)\\\widehat\vartheta|_{L\oplus H_\bu}=\psi}}Q_{\bu,\vartheta}.
    \end{align}

We claim that the set $\mathfrak{W}=\{W_{L\oplus H_\bu,\psi}\}$ is a $1$-compatible perfect quantum assignment for $\rho_1 S$ over $\HH$.

First of all, since each family $\{Q_{\bu,\vartheta}\suchThat\vartheta\in s(\bu)\}$ is a PVM, the linear applications in~\eqref{eqn_0801_1740} are projectors. Take now two different labels $\psi\neq\psi':L\oplus H_\bu\to\F_2$. If $\vartheta,\vartheta'\in s(\bu)$ are such that $\widehat\vartheta|_{L\oplus H_\bu}=\psi$ and $\widehat\vartheta'|_{L\oplus H_\bu}=\psi'$, it must hold that $\vartheta\neq\vartheta'$, which means that $Q_{\bu,\vartheta}\cdot Q_{\bu,\vartheta'}=O$. It follows that $W_{L\oplus H_\bu,\psi}\cdot W_{L\oplus H_\bu,\psi'}=O$; i.e., the projectors in $\mathfrak{W}$ corresponding to the same variable are mutually orthogonal. Furthermore, we have
\begin{align*}
    \sum_{\psi\in A}W_{L\oplus H_\bu,\psi}
        =
    \sum_{\psi\in A}\sum_{\substack{\vartheta\in s(\bu)\\\widehat\vartheta|_{L\oplus H_\bu}=\psi}}Q_{\bu,\vartheta}
    =
    \sum_{\substack{\vartheta:\tilde \bu\to\F_2\\\widehat\vartheta\mbox{ respects }\bu}}Q_{\bu,\vartheta}
    =
    \sum_{\vartheta\in s(\bu)}Q_{\bu,\vartheta}
    =\id_\HH.
\end{align*}
This shows that each set $\{W_{L\oplus H_\bu,\psi}\suchThat \psi\in A\}$ is a PVM and, thus, $\mathfrak{W}$ is a quantum assignment. 

Consider now a pair $e=(L\oplus H_\bu,L'\oplus H_\bw)\in E$
(where $E$ is the edge set of $\rho_1 S$), and take a pair $(\psi,\psi')$ of labels, with $\psi:L\oplus H_\bu\to\F_2$ and $\psi':L'\oplus H_\bw\to\F_2$. Recall the definition of the linear maps $\psi^{(\bw)}:L\oplus H_\bu\oplus H_\bw\to\F_2$ and $\psi'^{(\bu)}:L'\oplus H_\bu\oplus H_\bw\to\F_2$ given in Section~\ref{subsec_description_rho_one}. Suppose that $(\psi,\psi')\not\in\phi_e$, which means that 
\begin{align}
\label{eqn_1851_0801}
    \psi^{(\bw)}|_{(L\oplus H_\bu\oplus H_\bw)\cap(L'\oplus H_\bu\oplus H_\bw)}\neq\psi'^{(\bu)}|_{(L\oplus H_\bu\oplus H_\bw)\cap(L'\oplus H_\bu\oplus H_\bw)}.
\end{align}

Consider two linear functions $\vartheta\in s(\bu),\vartheta'\in s(\bw)$ such that $\widehat{\vartheta}|_{L\oplus H_\bu}=\psi$ and $\widehat{\vartheta'}|_{L'\oplus H_\bw}=\psi'$. Suppose, for the sake of contradiction, that
\begin{align*}
    Q_{\bu,\vartheta}\cdot Q_{\bw,\vartheta'}\neq O.
\end{align*}
Since the quantum assignment $\mathfrak{Q}$ is perfect, it follows that $\vartheta|_{\widetilde\bu\cap\widetilde{\bw}}=\vartheta'|_{\widetilde\bu\cap\widetilde{\bw}}$. Hence, there exists a well defined linear map $\xi:\tilde\bu\cup\tilde\bw\to\F_2$ extending both $\vartheta$ and $\vartheta'$. Note that $\widehat\xi$ respects both $\bu$ and $\bw$. Hence, by Lemma~\ref{lem_extension_side_conditions} on the unicity of extension, it must hold that $\widehat\xi|_{L\oplus H_\bu\oplus H_\bw}=\psi^{(\bw)}$ and $\widehat\xi|_{L'\oplus H_\bu\oplus H_\bw}=\psi'^{(\bu)}$. But this means that the functions $\psi^{(\bw)}$ and $\psi'^{(\bu)}$ agree on the intersection of their domains, which contradicts~\eqref{eqn_1851_0801}. Hence, the product $Q_{\bu,\vartheta}\cdot Q_{\bw,\vartheta'}$ is zero.
We deduce that
\begin{align*}
    W_{L\oplus H_\bu,\psi}\cdot W_{L'\oplus H_\bw,\psi'}&=\Big(\sum_{\substack{\vartheta\in s(\bu)\\\widehat\vartheta|_{L\oplus H_\bu}=\psi}}Q_{\bu,\vartheta}\Big)\Big(\sum_{\substack{\vartheta'\in s(\bw)\\\widehat\vartheta'|_{L'\oplus H_\bw}=\psi'}}Q_{\bw,\vartheta'}\Big)\\
    &=
    \sum_{\substack{\vartheta\in s(\bu)\\\widehat\vartheta|_{L\oplus H_\bu}=\psi}}\;\sum_{\substack{\vartheta'\in s(\bw)\\\widehat\vartheta'|_{L'\oplus H_\bw}=\psi'}}Q_{\bu,\vartheta}\cdot Q_{\bw,\vartheta'}=O,
\end{align*}
as needed. We conclude that the quantum assignment $\mathfrak{W}$ is perfect. 

We are left to establish $1$-compatibility. Like above, take an edge $e=(L\oplus H_\bu,L'\oplus H_\bw)\in E$ and two labels $\psi:L\oplus H_\bu\to\F_2$ and $\psi':L'\oplus H_\bw\to\F_2$. 
Suppose, for the sake of contradiction, that $\tilde\bu\cap\tilde\bw=\emptyset$. It follows that $\F_2^{\tilde\bu}\cap\F_2^{\tilde\bw}=\{\bzero\}$. Recall that $L\oplus H_\bu\subseteq\F_2^{\tilde\bu}$ and $L'\oplus H_\bw\subseteq\F_2^{\tilde\bw}$. Moreover, from $L'\in\mathcal{L}_\bw$, we deduce that $L'\cap H_\bw=\{\bzero\}$. As a consequence, 
the two spaces $L'$ and $L\oplus H_\bu\oplus H_\bw$ are linearly independent, which means that 
\begin{align*}
\dim(L\oplus L'\oplus H_\bu\oplus H_\bw)=\ell+\dim(L\oplus H_\bu\oplus H_\bw).
\end{align*}
Since $\ell\geq 2$, this would imply that $e\not\in E$, a contradiction. Hence, the two sets $\tilde\bu$ and $\tilde \bw$ are not disjoint. Since the assignment $\mathfrak{Q}$ is $1$-compatible, this means that $[Q_{\bu,\vartheta},Q_{\bw,\vartheta'}]=O$ for each $\vartheta\in s(\bu)$, $\vartheta'\in s(\bw)$. Using the linearity of the commutator, we deduce that $[W_{L\oplus H_\bu,\psi},W_{L'\oplus H_\bw,\psi'}]=O$, as needed.

We now need to extend quantum completeness to the second part $\rho_2$ of the reduction $\rho$. Recall from the description in Section~\ref{subsec_description_rho_one} that $\rho_2$ simply discards some of the constraints of $\rho_1(S)$. Hence, it is immediate to see that the same assignment $\mathfrak{W}$ is also a perfect $1$-compatible quantum assignment for $\rho S=\rho_2\rho_1 S$, and the proof is concluded.
\end{proof}

\subsection{Classical soundness of $\rho$}
\label{subsec_classical_soundness_rho}
The core of the sequence of works~\cite{Khot17:stoc-independent,Khot18:focs-pseudorandom,Dinur18:stoc-non-optimality,Dinur18:stoc-towards} is a proof of the classical soundness of $\rho$, established via an advanced analysis of the expansion properties of the Grassmann graphs encoding the geometry of the linear spaces $L\oplus H_\bu$ involved in the reduction.
\begin{thm}[\cite{Khot18:focs-pseudorandom}]
\label{thm_soundess_first_reduction}
For each $0<s^*<1$ and each $\delta>0$, there exist parameters $n,\ell\in\N$ such that the corresponding reduction $\rho$ satisfies the following: 
If the $3$XOR instance $S$ is such that $\sat(S)< s^*$, the $2$-to-$2$ instance $\rho S$ 
is such that $\sat(\rho S)< \delta$.
%
\end{thm}

In order to link $\rho$ with the reductions of Section~\ref{sec_reductions} and, thus, use it for transferring pseudo-telepathy all the way to the colouring game, we need a different classical soundness than the one in Theorem~\ref{thm_soundess_first_reduction}, involving the $\isat$ value as opposed to the $\sat$ value (cf.~Section~\ref{subsec_preprocessing}). There exist multiple ways to circumvent this technical issue. One way is to slightly modify the reduction $\rho$ by making the constraints of $\rho S$ be $2$-to-$1$ instead of $2$-to-$2$, as is done in~\cite{Dinur18:stoc-towards} (see also~\cite[Remark~3.4.11]{minzer2022monotonicity}). This involves considering two different types of vertices for $\rho_1 S$ corresponding to $\ell$-dimensional and $(\ell-1)$-dimensional subspaces of $H_\bu$, respectively, and suitably modifying the constraints and the weights in $\rho_2$. As a result, we would obtain a weighted $2$-to-$1$ instance admitting a perfect quantum assignment and arbitrarily small classical $\sat$ value, which is ready to be plugged in the chain of reductions from Sections~\ref{subsec_preprocessing} and~\ref{subsec_GS_reduction}. 

A second fix, which does not involve modifying the reduction $\rho$, is to apply the general framework of Crescenzi, Silvestri, and Trevisan~\cite{crescenzi2001weighted} based on the construction of mixing sets, which allows reducing weighted to unweighted instances of a large class of combinatorial problems including  CSPs without unary constraints---and suitably adapting the argument of~\cite[Claim~24]{crescenzi2001weighted} to capture soundness expressed in terms of $\isat$.

However, we can avoid both complications by making use of a recent result from~\cite{DawarM25}, showing through a probabilistic argument that the two soundness requirements for $\rho S$ are in fact equivalent. (The $\isat$ value is called ``irregular value'' in~\cite{DawarM25}.)

\begin{lem}[\cite{DawarM25}]
\label{lem_classical_soundness_minzer_isat}
    Let $\Phi$ be a (weighted) $2$-to-$2$ instance with $\isat_t(\Phi)=\epsilon$ for some $t\in\N$ and $\epsilon>0$. Then $\sat(\Phi)=\Omega(\frac{\epsilon^2}{t^2})$.
\end{lem}

\subsection{Proving Theorem~\ref{thm_equations_to_colouring}}
We now have the right forms of classical soundness and quantum completeness for $\rho$ that are necessary to link the reduction with that of Section~\ref{subsec_GS_reduction}, thus transferring pseudo-telepathy from the strong version of $3$XOR games repetition to the colouring game.

\begin{proof}[Proof of Theorem~\ref{thm_equations_to_colouring}]
    We follow the same structure as the proof of Theorem~\ref{thm_main}.
    Take a number $0<s^*<1$ as in the statement of the theorem, and choose an integer $c\in\N$. We aim to find a graph having quantum chromatic number at most $4$ and classical chromatic number at least $c+1$. Let $\epsilon=\frac{1}{c}$, and choose $\epsilon'>0$, $t\in\N$ such that Theorem~\ref{thm_soundess_third_reduction} holds. Take $\delta=\Omega(\frac{\epsilon'^2}{t^2})$ witnessing Lemma~\ref{lem_classical_soundness_minzer_isat}. Choose integers $n,\ell\in\N$ such that Theorem~\ref{thm_soundess_first_reduction} applied to $s^*$ and $\delta$ holds. Finally, take a regular $3$XOR instance $S$ such that $\sat(S)<s^*$ but $\mathcal{G}(S,n)$ admits a perfect $1$-compatible quantum assignment, as per the statement of Theorem~\ref{thm_equations_to_colouring}.

    Let $\rho$ be the reduction of Section~\ref{subsec_description_rho_one} with parameters $n$ and $\ell$, and let $\eta$ be the reduction from Section~\ref{subsec_GS_reduction}.
    By Theorem~\ref{thm_quantum_completeness_rho_minzer}, the $2$-to-$2$ instance $\rho S$ admits a perfect $1$-compatible quantum assignment. Therefore, we deduce from Theorem~\ref{thm_quantum_completeness_dinur} that $\eta\rho S\leadsto^1\K_4$ (and, thus, $\eta\rho S\leadsto^0\K_4$).
    However, Theorem~\ref{thm_soundess_first_reduction} implies that $\sat(\rho S)<\delta$; by Lemma~\ref{lem_classical_soundness_minzer_isat}, it follows that $\isat_t(\rho S)<\epsilon'$. Hence, Theorem~\ref{thm_soundess_third_reduction} implies that $\eta\rho S$ has no independent sets of relative size $\epsilon$---i.e., the classical chromatic number of $\eta\rho S$ is larger than $\frac{1}{\epsilon}=c$.

    To conclude, we simply need to symmetrise the digraph $\eta\rho S$ by turning each directed edge into an undirected edge. As noted in the proof of Theorem~\ref{thm_main}, this leaves the quantum and classical chromatic numbers unaffected. Hence, the resulting undirected graph has the required properties.
\end{proof}

\begin{rem}
\label{rem_proving_conj}
   We observe that the reduction discussed in this section not only is quantum complete, but it also preserves the Hilbert space witnessing the existence of perfect quantum assignments---just like the reduction we used to prove Theorem~\ref{thm_main}, see~Remark~\ref{rem_higher_consistency_gap}. Furthermore, it is clear from the proof of Theorem~\ref{thm_equations_to_colouring} that the resulting graph of arbitrarily large classical chromatic number is
   quantum $4$-colourable via 
   a $1$-compatible assignment. Hence, the same argument as in~Remark~\ref{rem_higher_consistency_gap}, based on~\cite[Theorem~29]{simul_banakh}, implies that the dimension of the Hilbert spaces in the perfect quantum strategies for the $3$XOR games $\mathcal{G}(S,n)$ in the hypothesis of Theorem~\ref{thm_equations_to_colouring} must tend to infinity as the number $n$ of repetitions of the game grows. 



As noted in Section~\ref{subsec_quantum_chromatic_gap_3xor}, the chain of reductions in the proof of Theorem~\ref{thm_equations_to_colouring} skips the $2$-to-$1$ games step, and reduces directly from $3$XOR to $2$-to-$2$ games. 
    Nevertheless, the case $d=2$, $k=1$ of Conjecture~\ref{conj_d_to_1_khot_quantum} can be derived from
    the hypothesis of Theorem~\ref{thm_equations_to_colouring} by applying the same argument to the modified reduction from~\cite{Dinur18:stoc-towards} instead of $\rho$, as discussed in Section~\ref{subsec_classical_soundness_rho}.
\end{rem}

\appendix

\section{The Dinur--Mossel--Regev reduction is quantum complete}
\label{app_dinur_reductions}

In this appendix, we prove the following result.

\begin{prop*}
[Proposition~\ref{prop_d_to_d_after_dinur} restated]
Suppose Conjecture~\ref{conj_d_to_1_khot_quantum} holds for some $d\geq 2$. 
Then for each $\epsilon>0$ and each $k,t\in\N$ there exists a $d$-to-$d$ instance $\Phi$ (on some alphabet size $n$) such that $\Phi$ admits a perfect $k$-compatible quantum assignment but
$\isat_t(\Phi)<\epsilon$.
\end{prop*}

We first need to introduce some new notation. 
Given a (weighted) label-cover instance $\Phi=(X,E,A,\pi,\phi)$ and a vertex $x\in X$, we define the $x$'th marginal of the probability distribution $\pi$ as 
\begin{align*}
\pi_x=\sum_{\substack{\bx\in E\\x_1=x}}\pi(\bx). 
\end{align*}
Moreover, given a classical assignment $f$ for $\Phi$, we let $\sat_{f,x}(\Phi)$ denote the value of $f$ restricted to constraints incident to $x$; formally,
\begin{align*}
    \sat_{f,x}(\Phi)=
    \sum_{\substack{\bx\in E\\x_1=x}}\pi(\bx)\,\phi_{\bx}(f(\bx)).
\end{align*}

We now show that a chain of simple transformations from $d$-to-$1$ games to $d$-to-$d$ games
given
in~\cite{Dinur09:sicomp} preserves the presence of perfect locally compatible quantum assignments. (See also~\cite{khot2008vertex}, where a similar chain of transformations is used to prove that a strong form of the Unique Games Conjecture follows from the standard form of~\cite{Khot02stoc}.)

\begin{lem}
\label{lem_dinur_uno}
    Fix a constant $h\in\N$. There is an (efficient) procedure that, given as input a bipartite label-cover instance $\Phi$ on variable set $X_1\sqcup X_2$, outputs a bipartite label-cover instance $\Phi'$ on variable set $X_1'\sqcup X_2$ satisfying the following:
    \begin{enumerate}
        \item If $\Phi$ admits a perfect $k$-compatible quantum assignment then $\Phi'$ admits a perfect $k$-compatible quantum assignment;
        \item for any $\delta,\epsilon>0$, if there exists a classical assignment $f'$ for $\Phi'$ for which a $\delta$-fraction of the variables $x'\in X_1'$ satisfy $\sat_{f',x'}(\Phi')\geq \frac{\epsilon}{|X_1'|}$, then $\sat(\Phi)\geq (1-\frac{1}{h})\delta\epsilon$.
        \item $\pi'_{x'}$ is constant over $x'\in X_1'$.
    \end{enumerate}
\end{lem}
\begin{proof}
    Let $\Phi=(X_1\sqcup X_2,E,A,\pi,\phi)$ be the given bipartite label-cover instance. Following~\cite{Dinur09:sicomp}, we consider the bipartite label-cover instance $\Phi'=(X_1'\sqcup X_2,E',A,\pi',\phi')$ defined as follows. 
    For each $x\in X_1$, set $\alpha_x=\floor{h\cdot|X_1|\cdot\pi_x}$.
    The set $X_1'$ contains $\alpha_x$ copies of each $x\in X_1$, denoted by $x^{(1)},\dots,x^{(\alpha_x)}$. For each $x\in X_1$, $y\in X_2$, and $i\in [\alpha_x]$, we add the edge $(x^{(i)},y)$ in $E'$ if and only if $(x,y)\in E$. In such case, we let $\pi'(x^{(i)},y)=\frac{\pi(x,y)}{|X_1'|\pi_x}$ and $\phi'_{(x^{(i)},y)}=\phi_{(x,y)}$. Observe that

\begin{align*}
    \pi'_{x^{(i)}}=\sum_{\substack{\bx\in E\\x_1=x}}\frac{\pi(\bx)}{|X_1'|\pi_x}
    =
    \frac{1}{|X_1'|},
\end{align*}
so condition $\textit{3.}$ is satisfied. Moreover, condition $\textit{2.}$ was proved in~\cite[Lemma~A.4]{Dinur09:sicomp}.

To prove that condition $\textit{1.}$ holds as well,
suppose that $\Phi$ admits a perfect $k$-compatible quantum assignment $\mathfrak{Q}=\{Q_{x,a}:x\in X_1\sqcup X_2,a\in A\}$. Consider the set $\mathfrak{W}=\{W_{ x',a}:x'\in X_1'\sqcup X_2,a\in A\}$ defined as follows. If $x'\in X_1'$ and, thus, $x'=x^{(i)}$ for some $x\in X_1$ and $i\in [\alpha_x]$, we let $W_{x',a}=Q_{x,a}$; if $x'\in X_2$, we let $W_{x',a}=Q_{x,a}$. It is straightforward to check that $\mathfrak{W}$ is a perfect quantum assignment for $\Phi'$.
To show that $\mathfrak{W}$ is $k$-compatible, assign to each vertex $y\in X_1'$ the vertex $\bar y\in X_1$ such that $y=\bar y^{(i)}$ for some $i\in[\alpha_{\bar y}]$, and let $\bar{y}=y$ for each $y\in X_2$. It is clear that $\dist_{\Phi'}(y_1,y_2)=\dist_\Phi(\bar y_1,\bar y_2)$ for each pair of vertices $y_1,y_2\in X_1'\sqcup X_2$ such that $\bar y_1\neq \bar y_2$. Hence, using the $k$-compatibility of $\mathfrak{Q}$, we find that for each  $a_1,a_2\in A$ and each $y_1,y_2\in X_1'\sqcup X_2$ with
$\dist_{\Phi'}(y_1,y_2)\leq k$ it holds that
\begin{align*}
[W_{y_1,a_1},W_{y_2,a_2}]=[Q_{\bar y_1,a_1},Q_{\bar y_2,a_2}]=O
\end{align*}
whenever
$\bar y_1\neq \bar y_2$. Furthermore, if $\bar y_1= \bar y_2$, we obtain
\begin{align*}
[W_{y_1,a_1},W_{y_2,a_2}]=[Q_{\bar y_1,a_1},Q_{\bar y_2,a_2}]= O
\end{align*}
as the projectors $Q_{\bar y_1,a_1}$ and $Q_{\bar y_2,a_2}$ are orthogonal in this case.
It follows that $\mathfrak{W}$ is $k$-compatible, as required. 
\end{proof}

We say that a bipartite label-cover instance on bipartition $X_1\sqcup X_2$ is \textit{left regular} if each vertex in $X_1$ has the same number of incident edges.

\begin{lem}
\label{lem_dinur_due}
    Fix two constants $\ell,m\in\N$. There is an (efficient) procedure that, given as input a $d$-to-$1$ instance $\Phi$ on variable set $X_1\sqcup 
    X_2$ such that $\pi_x$ is constant over $x\in X_1$, outputs a left-regular unweighted $d$-to-$1$ instance $\Phi'$ satisfying the following:
    \begin{enumerate}
    \item If $\Phi$ admits a perfect $k$-compatible quantum assignment then $\Phi'$ admits a perfect $k$-compatible quantum assignment;
    \item for any $\delta>0$ and any $0<\epsilon<\frac{1}{d\ell^2}$, if $\sat(\Phi')>\delta+\frac{1}{\ell}+(1-\delta)d\ell^2\epsilon$, then there exists a classical assignment $f$ for $\Phi$ such that $\delta$-fraction of the variables $x\in X_1$ satisfy $\sat_{f,x}(\Phi)>\frac{1}{|X_1|}(\epsilon-\frac{1}{m})$.
    \end{enumerate}
\end{lem}

\begin{proof}
    Let $\Phi=(X_1\sqcup X_2,E,A,\pi,\phi)$ be the given (weighted) $d$-to-$1$ instance, where $X_1\sqcup X_2$ is the bipartition of the vertex set.
    Let $\alpha=m\cdot|X_2|$.
    We construct a bipartite multigraph $\tilde G$ on vertex set $X_1\sqcup X_2$ and edge set $\tilde E$ defined as follows. 
    For each $x\in X_1$, choose an edge $\bar\bx\in E$ such that $\bar x_1=x$.
    For each edge $\bx\in E$ such that $x_1=x$ and $\bx\neq\bar\bx$, $\tilde E$ contains $\floor{\alpha\cdot\pi(\bx)}$ copies of $\bx$. Moreover, $\tilde E$ contains 
    \begin{align*}
        \alpha-\sum_{\substack{\bx\in E\setminus\{\bar\bx\}\\x_1=x}}\floor{\alpha\cdot\pi(\bx)}
    \end{align*}
    copies of the edge $\bar\bx$.
Observe that all vertices of $X_1$ have degree $\alpha$ in $\tilde G$. Thus, for $x\in X_1$, we let $\by^{(x)}=(y^{(x)}_1,\dots,y^{(x)}_\alpha)$ be the list of neighbours of $x$ in $\tilde G$, counted with the correct multiplicity.

Following~\cite{Dinur09:sicomp}, we define the unweighted $d$-to-$1$ instance $\Phi'=(X_1'\sqcup X_2,E',A,\phi')$ as follows.
    The set $X_1'$ contains a variable $(x,\bi)$ for each $x\in X_1$ and each $\ell$-tuple $\bi=(i_1,\dots,i_\ell)\in [\alpha]^\ell$.
For $x\in X_1$, $\bi\in [\alpha]^\ell$, and $j\in[\ell]$, we add to $E'$ the edge $\bx'=((x,\bi),y^{(x)}_{i_j})$. The set of admitted labellings for $\bx'$ is $\phi'_{\bx'}=\phi_{(x,y^{(x)}_{i_j})}$---in particular, $\Phi'$ is a $d$-to-$1$ instance.

The fact that $\Phi'$ satisfies condition $\textit{2.}$ was proved in~\cite[Lemmas~A.5--A.7]{Dinur09:sicomp}.
To prove condition $\textit{1.}$,
suppose that $\Phi$ admits a perfect $k$-compatible quantum assignment $\mathfrak{Q}=\{Q_{x,a}:\, x\in X_1\sqcup X_2, a\in A\}$.
%
It is straightforward to check that the set $\mathfrak{W}=\{W_{x',a}:\, x'\in X_1'\sqcup X_2, a\in A\}$ defined by $W_{(x,\bi),a}=Q_{x,a}$ for each $(x,\bi)\in X_1'$, $a\in A$ and $W_{x,a}=Q_{x,a}$ for each $x\in X_2$, $a\in A$ yields a perfect $k$-compatible quantum assignment for $\Phi'$, as needed.
\end{proof}

%
Before proceeding with the last step of the transformation from $d$-to-$1$ to $d$-to-$d$ games, 
we need to prove the following technical fact, stating that the PVMs in an optimal quantum strategy for a bipartite label-cover instance can be assumed, without loss of generality, to project onto the trivial space $\{\bzero\}$ unless the corresponding variable and label are in matching parts of the bipartition of the instance.
\begin{prop}
\label{prop_only_relevant_projectors}
    Let $\Phi=(X_1\sqcup X_2,E,A,\pi,\phi)$ be a bipartite label-cover instance and let $A=A_1\sqcup A_2$ be a decomposition of $A$ such that $\phi_\bx\subseteq A_1\times A_2$ for each $\bx\in E$.
    Given any $k$-compatible quantum assignment $\mathfrak{Q}$ for $\Phi$, there exists a $k$-compatible quantum assignment $\mathfrak{W}=\{W_{x,a}:\, x\in X_1\sqcup X_2, a\in A\}$ for $\Phi$ such that
    \begin{itemize}
        \item $\qsat_{\mathfrak{W}}(\Phi)\geq \qsat_{\mathfrak{Q}}(\Phi)$;
        \item $W_{x,a}=O$ whenever $x\in X_i$ and $a\in A_j$ for $i\neq j\in [2]$.
    \end{itemize}
\end{prop}
\begin{proof}
Take a $k$-compatible quantum assignment $\mathfrak{Q}=\{Q_{x,a}:\, x\in X_1\sqcup X_2, a\in A\}$, and 
consider the set 
    \begin{align*}
        S_{\mathfrak{Q}}
        =
        \{(x,a)\suchThat x\in X_i,\,a\in A_j,\, i\neq j,\, \mbox{and } Q_{x,a}\neq O\}.
    \end{align*}
If $S_{\mathfrak{Q}}=\emptyset$, there is nothing to prove. Otherwise, choose $(\tilde x,\tilde a)\in S_{\mathfrak{Q}}$ and fix a label $\hat a\in A_i$. Define the following linear applications: $W_{\tilde x,\tilde a}=O$, $W_{\tilde x,\hat a}=Q_{\tilde x,\tilde a}+Q_{\tilde x,\hat a}$, and $W_{x,a}=Q_{x,a}$ for all other choices of $x\in X_1\sqcup X_2$ and $a\in A$. Consider now the set $\mathfrak{W}=\{W_{x,a}:\, x\in X_1\sqcup X_2, a\in A\}$.

It is straightforward to check that $\mathfrak{W}$ is a $k$-compatible quantum assignment for $\Phi$.
Furthermore, for each $\bx\in E$ and each $\ba\in A^2$, consider the linear application $M_{\bx,\ba}=W_{x_1,a_1}W_{x_2,a_2}-Q_{x_1,a_1}Q_{x_2,a_2}$.
We have
\begin{align}
\label{eqn_2612_1746}
    \qsat_{\mathfrak{W}}(\Phi)-\qsat_{\mathfrak{Q}}(\Phi)
    =
    \frac{1}{\dim(\HH)}\mathbb{E}_{\bx\sim\pi}
    \sum_{\ba\in \phi_\bx}\tr(M_{\bx,\ba}).
\end{align}
We claim that $\tr(M_{\bx,\ba})\geq 0$ for each $\bx\in E$ and each $\ba\in\phi_\bx$. Suppose for concreteness that $i=1$ and $j=2$, so $\tilde x\in X_1$, $\tilde a\in A_2$, and $\hat a\in A_1$ (the other case is analogous).
If $x_1\neq \tilde x$ or $a_1\not\in\{\tilde a,\hat a\}$, then $M_{\bx,\ba}=O$ and the claim is true, so we can assume that $x_1=\tilde x$ and $a_1\in\{\tilde a,\hat a\}$. 
Since $\ba\in\phi_\bx\subseteq A_1\times A_2$, it must hold that $a_1\in A_1$, so $a_1\neq\tilde a$. The only remaining case is when $x_1=\tilde x$ and $a_1=\hat a$. We obtain
\begin{align*}
    M_{\bx,\ba}
    =
    W_{\tilde x,\hat a}W_{x_2,a_2}-Q_{\tilde x,\hat a}Q_{x_2,a_2}
    =
    (Q_{\tilde x,\tilde a}+Q_{\tilde x,\hat a})Q_{x_2,a_2}-Q_{\tilde x,\hat a}Q_{x_2,a_2}
    =
    Q_{\tilde x,\tilde a}Q_{x_2,a_2}.
\end{align*}
Since the trace of the product of two positive semidefinite matrices is nonnegative, it follows that $\tr(M_{\bx,\ba})\geq 0$, as required. As a consequence,~\eqref{eqn_2612_1746} implies that $\qsat_{\mathfrak{W}}(\Phi)\geq\qsat_{\mathfrak{Q}}(\Phi)$.

Observe now that $|S_{\mathfrak{W}}|=|S_{\mathfrak{Q}}|-1$. Hence, after repeating the process for a finite number of times, we end up with a $k$-compatible quantum assignment satisfying the conditions of the proposition.
\end{proof}

\begin{lem}
\label{lem_dinur_tre}
There is an (efficient) procedure that, given as input an unweighted left-regular $d$-to-$1$ instance $\Phi$, outputs an
unweighted $d$-to-$d$ instance $\Phi'$ satisfying the following:
    \begin{enumerate}
        \item If $\Phi$ admits a perfect $2k$-compatible quantum assignment then $\Phi'$ admits a perfect $k$-compatible quantum assignment;
        \item for any $\epsilon>0$ and any $t\in\N$, if $\isat_t(\Phi')\geq\epsilon$, then $\sat(\Phi)\geq\frac{\epsilon}{t^2}$.
    \end{enumerate}
\end{lem}
\begin{proof}
Let $\Phi=(X_1\sqcup X_2,E,A,\phi)$ be the given $d$-to-$1$ instance. 
Recall from the definition of $d$-to-$1$ games in Section~\ref{subsec_CSP_predicates} that the alphabet $A$ can be decomposed as $A=A_1\sqcup A_2$ in a way that $\phi_\bx\subseteq A_1\times A_2$ for each $\bx\in E$.
Following~\cite{Dinur09:sicomp}, we consider the unweighted label-cover instance $\Phi'=(X_1,E',A_1,\phi')$ defined as follows. For each pair $(x,y)$, $(x',y)$ of edges in $E$ sharing a variable in $X_2$, we add the edge $(x,x')$ in $E'$. The corresponding set of admitted labellings is
    \begin{align*}
        \phi'_{(x,x')}=\{(a,a')\in A_1^2\suchThat\exists b\in A_2:\,(a,b)\in\phi_{(x,y)},\,(a',b)\in\phi_{(x',y)}\}.
    \end{align*}
It is easy to check that all constraints of $\Phi'$ are $d$-to-$d$, and it was shown in~\cite[Lemma~A.9]{Dinur09:sicomp} that condition $2.$ holds.

We now show that condition $1.$ holds, too. 
Take a perfect $2k$-compatible quantum assignment $\mathfrak{Q}=\{Q_{x,a}:\, x\in X_1\sqcup X_2, a\in A\}$ for $\Phi$ over some Hilbert space $\HH$.
Using Proposition~\ref{prop_only_relevant_projectors}, we can assume without loss of generality that $Q_{x,a}=O$ whenever $x\in X_i$ and $a\in A_j$ for $i\neq j\in [2]$.
We define $W_{x,a}=Q_{x,a}$ for each $x\in X_1$ and each $a\in A_1$. Note that 
    \begin{align*}
        \sum_{a\in A_1}W_{x,a}
        =
        \sum_{a\in A_1}Q_{x,a}
        =
        \sum_{a\in A}Q_{x,a}
        =
        \id_\HH,
    \end{align*}
    where the second equality comes from the second condition in Proposition~\ref{prop_only_relevant_projectors}.
    Hence, the set $\mathfrak{W}=\{W_{x,a}\suchThat x\in X_1,a\in A_1\}$ is a quantum assignment for $\Phi'$. 
    
    Consider now an edge $(x,x')\in E'$, and let $y$ be the corresponding vertex of $X_2$ such that $(x,y)$ and $(x',y)$ are both in $E$. Take
    a pair $(a,a')\in A_1^2\setminus \phi'_{(x,x')}$, which means that, for each $b\in A_2$, either $(a,b)\not\in\phi_{(x,y)}$ or $(a',b)\not\in\phi_{(x',y)}$. 
    Since the quantum assignment $\mathfrak{Q}$ is perfect, for each $b\in A_2$ it holds that 
    \begin{align*}
        Q_{x,a}\cdot Q_{y,b}\neq O
        \quad\Rightarrow \quad
        (a,b)\in\phi_{(x,y)}
        \quad\Rightarrow \quad
        (a',b)\not\in\phi_{(x',y)}
        \quad\Rightarrow \quad
        Q_{x',a'}\cdot Q_{y,b}= O,
    \end{align*}
    which means that 
    \begin{align*}
    Q_{x,a}\cdot Q_{y,b}\cdot Q_{x',a'}\cdot Q_{y,b}=O.
    \end{align*}
    Since that the vertices $x'$ and $y$ are adjacent in $\Phi$, the corresponding projectors of $\mathfrak{Q}$ commute, so we obtain
    \begin{align*}
        Q_{x,a}\cdot Q_{x',a'}\cdot Q_{y,b}=O
    \end{align*}
    for each $b\in A_2$. Taking the sum over $A_2$ (and using again the second condition of Proposition~\ref{prop_only_relevant_projectors}), we find
    \begin{align*}
        O&=\sum_{b\in A_2}Q_{x,a}\cdot Q_{x',a'}\cdot Q_{y,b}
        =
        Q_{x,a}\cdot Q_{x',a'}\cdot \sum_{b\in A_2}Q_{y,b}
        =
        Q_{x,a}\cdot Q_{x',a'}\cdot \id_\HH\\
        &=
        Q_{x,a}\cdot Q_{x',a'}
        =
        W_{x,a}\cdot W_{x',a'},
    \end{align*}
    whence we deduce that the quantum assignment $\mathfrak{W}$ is perfect.

Finally, it is clear from the definition of $\Phi'$ that $\dist_{\Phi'}(x,x')=\frac{1}{2}\cdot\dist_{\Phi}(x,x')$ for each pair of vertices $x,x'\in X_1$. Hence, the $k$-compatibility of $\mathfrak{W}$ follows directly from the $2k$-compatibility of $\mathfrak{Q}$, thus concluding the proof.
\end{proof}

\begin{proof}[Proof of Proposition~\ref{prop_d_to_d_after_dinur}]

Fix $d,\epsilon, k,t$ as in the statement of the proposition, and consider the parameters $\delta=\frac{\epsilon}{3t^2}$, $\epsilon'=\frac{\epsilon}{3d\ceil{1/\delta}^2t^2}$, $\epsilon''=\frac{\epsilon'}{2}$, and $\epsilon'''=\frac{\delta\epsilon''}{2}$.

Using Conjecture~\ref{conj_d_to_1_khot_quantum}, let $\Phi$ be a $d$-to-$1$ instance such that 
$\Phi$ admits a perfect $2k$-compatible quantum assignment
but $\sat(\Phi)<\epsilon'''$. 
We now consecutively apply to $\Phi$ the three transformations of Lemmas~\ref{lem_dinur_uno},~\ref{lem_dinur_due}, and~\ref{lem_dinur_tre}.
Note that the transformation in Lemma~\ref{lem_dinur_uno} does not change the constraints in the input label-cover instance. Hence, it maps $d$-to-$1$ instances to $d$-to-$1$ instances. Let $\Phi'$ be the result of applying Lemma~\ref{lem_dinur_uno} to $\Phi$ with parameter $h=2$; let $\Phi''$ be the result of applying Lemma~\ref{lem_dinur_due} to $\Phi'$ with parameters $\ell=\ceil{1/\delta}$ and $m=\ceil{2/\epsilon'}$; and let $\Phi'''$ be the result of applying Lemma~\ref{lem_dinur_tre} to $\Phi''$. We claim that the resulting $d$-to-$d$ instance $\Phi'''$ meets the requirements of the proposition. First, the quantum completeness of the transformations implies that $\Phi$, $\Phi'$, and $\Phi''$ admit perfect $2k$-compatible quantum assignments and, thus, $\Phi'''$ admits a perfect $k$-compatible quantum assignment.


Observe now that $(1-\frac{1}{h})\delta\epsilon''=\epsilon'''>\sat(\Phi)$.
Hence, part $\textit{2.}$ of Lemma~\ref{lem_dinur_uno} (with parameters $\delta,\epsilon''$) implies that there is no assignment $f'$ for $\Phi'$ for which a $\delta$-fraction of the variables $x'\in X_1'$ satisfy $\sat_{f',x'}(\Phi')\geq\frac{\epsilon''}{|X_1'|}$, where $X_1'$ is the left-part of the variable set of $\Phi'$. Next, note that the parameter choice guarantees that $0<\epsilon'<\frac{1}{d\ell^2}$ and $\epsilon'-\frac{1}{m}\geq \epsilon''$. Hence, part $\textit{2.}$ of Lemma~\ref{lem_dinur_due} (with parameters $\delta,\epsilon'$) guarantees that $\sat(\Phi'')\leq \delta+\frac{1}{\ell}+(1-\delta)d\ell^2\epsilon'$. Finally, observe that
\begin{align*}
    \delta+\frac{1}{\ell}+(1-\delta)d\ell^2\epsilon'
    <
    \delta+\frac{1}{\ell}+d\ell^2\epsilon'
    \leq
    \frac{\epsilon}{3t^2}+\frac{\epsilon}{3t^2}+\frac{\epsilon}{3t^2}=\frac{\epsilon}{t^2}.
\end{align*}
Therefore, part $\textit{2.}$ of Lemma~\ref{lem_dinur_tre} implies that $\isat_t(\Phi''')<\epsilon$, thus concluding the proof of the proposition.
\end{proof}

\begin{rem}
\label{rem_appendix_reductions_preserve_H}
We point out that the Hilbert space associated with quantum assignments is preserved throughout the transformations occurring in the proof of Proposition~\ref{prop_d_to_d_after_dinur}. Indeed, the proofs of Lemmas~\ref{lem_dinur_uno},~\ref{lem_dinur_due}, and~\ref{lem_dinur_tre} and Proposition~\ref{prop_only_relevant_projectors} show how to transfer perfect locally compatible quantum assignments for a given $d$-to-$1$ instance over a Hilbert space $\HH$ to perfect locally compatible quantum assignments for the resulting $d$-to-$d$ instance over the \textit{same} space $\HH$.
\end{rem}

\section*{Acknowledgements}
I thank Samson Abramsky, Andrea Coladangelo, and Amin Karamlou for insightful discussions and feedback on various parts of this work, which significantly contributed to improving the final version of the manuscript. The research leading to these results was supported by UKRI EP/X024431/1.

{\small
\bibliographystyle{alphaurl}
\bibliography{bibliography}
}

\end{document}